\pdfoutput=1
\documentclass[11pt,reqno]{amsart}

\usepackage{amsmath}
\usepackage{amssymb}
\usepackage{amsaddr}
\usepackage{amsthm}
\usepackage{array}
\usepackage{booktabs}
\usepackage{bm}
\usepackage{caption}
\usepackage{enumerate}
\usepackage{flafter}
\usepackage{geometry}
\usepackage{graphicx}
\usepackage{hyperref}
\usepackage{multirow}
\usepackage{natbib}
\usepackage{subcaption}
\usepackage[svgnames]{xcolor}
\usepackage{stmaryrd}
\usepackage{thmtools}

\graphicspath{{./images/}}
\newtheorem{theorem}{Theorem}
\theoremstyle{plain}

\newtheorem{algorithm}{Algorithm}

\newtheorem*{assumption*}{Assumption}
\newtheorem{assumption}{Assumption}

\newtheorem{corollary}{Corollary}

\newtheorem{definition}{Definition}

\newtheorem{proposition}{Proposition}

\numberwithin{equation}{section}

\begin{document}

\title[The Markup Falsification Adaptive Set]{The Markup Falsification Adaptive Set}
\author{Santiago Acerenza, Néstor Gandelman}
\address{Universidad ORT Uruguay}
\date{\scriptsize The present version is as of \today.}

\begin{abstract}
In this paper we provide a constructive way for researchers to salvage the classic \cite{loecker2012markups} markup recovery procedure when falsified. To do this, we consider continuous relaxations of the standard assumptions behind markup estimation. By computing the values of the markup as a function of the relaxations across the set of non-falsified models, we obtain an identified set for the markup which generalizes the standard baseline markup estimand to account for possible falsification without the need to impose additional assumptions. We illustrate our results using Chilean data from \cite{raval2023testing}. We also provide an inferential procedure for the FAS endpoints based on their local differentiability structure.
\end{abstract}

\maketitle

{\footnotesize
\textbf{Keywords}: markup, falsification, sensitivity analysis, partial identification.

\textbf{JEL subject classification}: C18; C26; C51; E23; L0; L11.
}

\section{Introduction}

Markup estimation aims to recover the gap between prices and marginal costs at the firm or product level, a key object for studying market power. Empirical work on firm-level market power has developed around two primary methodological approaches. The demand approach \citep{berry1994estimating,berry1993automobile} relies on price and quantity data to infer marginal costs under alternative assumptions about firms’ profit-maximizing behavior. In contrast, the production approach does not require specifying demand systems or detailed competitive interactions. This method, which traces back to \cite{hall1988relation}, was formalized and brought to prominence by \cite{loecker2012markups}.

Within the production approach, markups are inferred from firms’ cost-minimization conditions. Under standard assumptions, the markup can be expressed as the ratio of the output elasticity of a variable input to that input’s revenue share. This approach has become widely used because it avoids the need to estimate demand systems and instead relies on the consistent estimation of production functions and input elasticities, making it particularly attractive in settings with detailed firm-level data.

When more than one variable input is available, the assumptions required to recover markups imply a set of testable restrictions on observable data. This feature renders the model empirically falsifiable, as these restrictions can be directly assessed in the data. More generally, model falsifiability refers to the property of a theoretical or structural model whereby it generates testable restrictions on observable data that could, in principle, be contradicted. A model is falsifiable if its assumptions imply relationships—such as moment conditions, inequalities, or functional forms—that must hold in the data; failure of these implications constitutes evidence against the model. In empirical economics, falsifiability is crucial because it disciplines model selection: rather than fitting any pattern ex post, a falsifiable model exposes itself to rejection through observable implications. When a model is falsified, it indicates that at least one of its underlying assumptions is inconsistent with the data, motivating either model revision, relaxation of assumptions, or a shift toward frameworks that allow for partial identification.

With finite samples, researchers often use specification tests to check whether their baseline model is false. This is the case for the production approach to markup estimation as noted by \cite{raval2023testing}. Abstracting from sampling uncertainty, the population versions of such specification tests have a persistent problem: what should researchers do when markups differ between variable inputs? Or, more abstractly, what should researchers do when the baseline markup model is falsified?

Building on the framework of \cite{masten2021salvaging}, this paper provides a constructive framework for researchers to evaluate and recover markup estimates when the underlying structural assumptions of the production approach are empirically falsified. Concurrently, our approach yields an indirect specification test for the baseline model itself.

To achieve this, we introduce continuous relaxations of two fundamental baseline assumptions: the assumption of cost minimization in the absence of adjustment costs and the correct specification of output elasticities. By systematically weakening these restrictions, an otherwise falsified baseline model can be rendered non-falsified. This allows us to define the \textit{falsification frontier}, which characterizes the minimal set of relaxations required to ensure the model remains compatible with the observed data. Under the assumption that the true structural parameters lie on this boundary, we estimate the \textit{falsification adaptive set} (FAS) for markups. The FAS generalizes the standard baseline markup estimand to explicitly account for model misspecification, crucially eliminating the need to impose restrictive alternative assumptions on the production function or cost minimization.

When interpreting these results, two core tenets of the falsification literature must be kept in mind. First, while a falsified model is logically impossible, a non-falsified model is not inherently true; it is merely rationalizable by the data. Second, when the baseline model is falsified, the rejection may stem from any combination of the underlying assumptions. Importantly, this limitation holds regardless of the geometric shape of the falsification frontier: the frontier maps the magnitude of the required violations but cannot isolate the precise causal mechanism behind the model’s structural failure.

Our work contributes to the recent markup-estimation literature by providing a falsification-based framework for the production approach. While prior studies have noted that the identifying assumptions behind production-based markups can fail, the literature has largely responded by imposing alternative structural assumptions. In contrast, we treat empirical refutation as an input to identification. We characterize the minimal relaxations of the baseline assumptions needed to rationalize the data (the falsification frontier) and, under the assumption that the true model lies on this frontier, derive the falsification adaptive set (FAS) for markups. The FAS delivers an identified set for the true markup without requiring researchers to commit to additional modeling choices.

More broadly, our paper contributes to the falsification and sensitivity-analysis literature initiated by \citet{masten2021salvaging}. That literature advocates reporting conclusions that remain consistent with minimally non-falsified models and has been extended in several directions, including computational partial identification for policy-relevant parameters \citep{HanYang2024PolicyEval}, minimax-robust inference under local misspecification \citep{BonhommeWeidner2022MinSens}, and analyses of how alternative relaxations of refuted set-identified models can lead to conflicting outer sets \citep{LiKedagniMourifie2024Discordant}. Closely related work generalizes the falsification-adaptive-set idea to linear IV settings with potentially invalid instruments \citep{ApfelWindmeijer2024GeneralizedFAS}. Complementary contributions develop sensitivity analysis for approximate moment-condition models \citep{ArmstrongKolesar2021ApproxMoment}, specification testing tools for partially identified models with many conditional moment inequalities \citep{MarcouxRussellWan2024SpecTest}, and breakdown-style analyses for IV models with binary outcomes \citep{Picchetti2025BreakdownIVBinary}. Our setting brings these ideas to production-based markup estimation and shows how falsification can be used constructively to deliver both identification and diagnostics.

At the inferential level, the FAS endpoints are globally nonsmooth but locally affine away from regime-switching boundaries. This local structure allows us to use the framework of \cite{fang2019inference} to construct valid confidence intervals for each scalar endpoint of the FAS.

The remainder of the paper is organized as follows. Section \ref{SecIdentification} reviews the production approach to markup estimation. Section \ref{Why} discusses why two key assumptions underlying the model may fail. Section \ref{SecFalsification} presents our main results, including the relaxation of the baseline assumptions and the FAS for markup models. Section \ref{SecIllustration} illustrates the results using data from \cite{raval2023testing}. Finally, Section \ref{SecConclusions} concludes. Proofs of the results are provided in Appendix \ref{Appproofs}.

\section{Markup production approach framework}\label{SecIdentification}

This section derives the expression for product markups following \cite{loecker2012markups}. Consider a firm $i$ in period $t$ with production technology
\[
Q_{it}=Q_{it}(V^{1}_{it},V^{2}_{it},K_{it},\omega_{it}),
\]
where $V^{1}_{it}$ and $V^{2}_{it}$ denote two variable inputs, $K_{it}$ denotes capital, $\omega_{it}$ is a scalar productivity term, and $Q_{it}$ denotes gross output.

\begin{assumption}\label{costmin}
Firms minimize costs in the absence of adjustment costs.
\end{assumption}

Under Assumption \ref{costmin}, the first-order condition for any variable input must be
\begin{equation}
\label{eq:foc}
P^{j}_{it} = \lambda_{it}\frac{\partial Q_{it}}{\partial V^{j}_{it}},
\end{equation}
where $P^{j}_{it}$ is the price of the variable input $V^j_{it}$ and $\lambda_{it}$ is the Lagrange multiplier, interpreted as the marginal cost of production.

Denoting the elasticity of output with respect to intermediate input $V^{j}_{it}$ as $\theta^{j}_{it}$,
\[
\theta_{it}^{j}
\equiv
\frac{\partial Q_{it}}{\partial V_{it}^j} \frac{V_{it}^j}{Q_{it}},
\]
and denoting the share of expenditures of input $V^j_{it}$ in total sales as $\alpha^{j}_{it}$,
\[
\alpha_{it}^{j}
\equiv
\frac{P_{it}^{j}V_{it}^{j}}{P_{it}Q_{it}},
\]
it follows that the markup, defined as the price-marginal-cost ratio, must equal
\begin{equation}
\label{eq:markup}
\mu_{it}^{j} = \frac{\theta_{it}^{j}}{\alpha_{it}^{j}}.
\end{equation}

In many firm-level datasets, expenditures on each intermediate input and total revenues are directly observed, allowing the corresponding input shares to be computed. Therefore, to estimate markups we only need a consistent estimator of $\theta_{it}^{j}$.

\begin{assumption}\label{elasticity}
There exists a procedure that allows unbiased recovery of the output-input elasticity:
\[
\theta^{j}_{it} = \hat{\theta}^{j}_{it}.
\]
\end{assumption}

Under Assumptions \ref{costmin} and \ref{elasticity}, a firm’s markup can be identified using variable input $j$, for $j=1,2$, as
\begin{equation}\label{EqMarkup1}
\hat{\mu}_{it}^{j}= \frac{\hat{\theta}_{it}^{j}}{\alpha_{it}^{j}}.
\end{equation}

In our context, the model would be falsified if, for the observed distribution of
\[
P_{it}, Q_{it}, V_{it}^{1}, V_{it}^{2},K_{it},P_{it}^{1}, P_{it}^{2}, P_{it}^{K},
\]
the quantities
\[
\frac{\hat{\theta}_{it}^{1}}{\alpha_{it}^{1}}
\qquad\text{and}\qquad
\frac{\hat{\theta}_{it}^{2}}{\alpha_{it}^{2}}
\]
are not equal for every $i$ and $t$. In the next section we discuss why this may occur.

\section{Why the baseline model may fail?}\label{Why}

First, if there are adjustment costs or firms do not minimize costs exactly, then the first-order condition \eqref{eq:foc} may fail, and therefore the markup may differ from \eqref{eq:markup}. These concerns have been extensively discussed in the literature.

\cite{hamermesh1996adjustment} notes that labor demand often involves adjustment costs due to changes in the identity of the individuals filling a fixed number of jobs. Such costs may include advertising vacancies, screening candidates, processing new employees, training, severance pay, and the overhead costs of maintaining recruitment and separation activities. Similarly, changes in the capital stock or in its rate of utilization can disrupt existing production plans. For example, delivery and installation of new equipment require time and internal reallocation of resources, while workers must learn how to operate the new capital.

Similar considerations may apply to intermediate inputs. Although they are often modeled as flexible margins of adjustment, modifying their use may also entail non-negligible costs. Adjusting input sourcing may require searching for and evaluating suppliers, renegotiating contracts, adapting logistics and inventory management, and coping with delivery lags or minimum order requirements. In production environments with technological complementarities or quality-specific inputs, changing the mix or intensity of intermediate inputs may also require modifications in routines, coordination across production stages, and learning about the performance of new materials or suppliers.

On the other hand, even if there are no adjustment costs, firm decision makers may exhibit bounded rationality, a notion that goes back at least to \cite{simon1955behavioral}, who emphasized that agents may satisfice rather than solve fully rational optimization problems. In this spirit, evidence from \cite{kelley2014experimental,morales2020bounded} suggests that firms may depart from frictionless cost minimization because decision makers face cognitive and computational limitations. As noted by \cite{sterman2007getting}, boundedly rational managers may fail to anticipate market saturation in time to reduce capacity. Additionally, \cite{conlisk1996bounded} argues that under bounded rationality a firm may be unable to compute, costlessly and exactly, its optimal output.

More recent work has formalized related departures from full optimization through rational inattention, where firms face information-processing constraints and therefore allocate limited attention imperfectly across decision margins \citep{sims2003implications,mackowiak2009optimal}. Likewise, the behavioral IO literature has emphasized that firms may make systematic mistakes, rely on heuristics, or use imperfect decision rules rather than behave as fully rational profit maximizers \citep{ellison2006bounded}. In these settings, even in the absence of physical or organizational adjustment costs, factor demand may become inconsistent with standard static cost minimization.

The second assumption in the previous section (correct specification of output elasticities) may also fail. This issue is closely related to production function estimation, as reviewed by \cite{de2021industrial}. Their survey covers the most widely used methods for estimating total factor productivity and recovering output elasticities. They distinguish between two broad approaches. The first is the factor-share approach, which recovers elasticities from the cost share of each input in total revenue. This approach requires all inputs to be flexible and the production technology to exhibit constant returns to scale. The second approach estimates output elasticities by directly fitting regressions of output on production inputs. This latter approach faces two well-known econometric challenges: simultaneity bias—arising because firms observe their productivity when choosing inputs—and selection bias—driven by the endogenous survival of more productive firms. The literature has addressed these problems through two main strategies: control function approaches and dynamic panel-data methods.\footnote{Dynamic panel-data methods, originally developed in a broader econometric context \citep{arellano1991some,blundell1998initial}, provide an alternative way to address these identification challenges.}

Recent work has highlighted important limitations in the identification of output elasticities. \cite{ackerberg2015identification} show that, under simple data-generating processes consistent with the standard assumptions, the moment conditions underlying the first-stage estimation of production-function parameters may fail to identify the labor coefficient, leading to misspecified elasticities. \cite{doraszelski2021reexamining} further argue that these methods are not robust to unobserved heterogeneity in demand across firms or over time. In a related contribution, they show, first, that under imperfect competition the estimation of input elasticities must account for markups themselves; and second, that under commonly used production-function specifications, cost minimization conditions are often inconsistent with observed data. Additional evidence on the limitations of existing procedures for elasticity recovery is provided by \cite{casacuberta2025use}, who show that in the presence of labor-market power the cost-share approach fails to identify the output elasticity of labor required to measure labor-market power. Finally, \cite{raval2023testing} test the implication that any flexible input should recover the same markup under the production approach and strongly reject the hypothesis that markups estimated using labor and materials share the same distribution.

In summary, within the context of the production approach to markup estimation, there are two fundamental sources of model failure under Assumptions \ref{costmin} and \ref{elasticity}:
\begin{enumerate}
    \item \textbf{Adjustment costs or failure of cost minimization.}
    This implies that although the production function is properly estimated and $\theta_{it}^{j}=\hat{\theta}_{it}^{j}$, the true markup differs from the ratio between the elasticity and the revenue share:
    \[
    \mu_{it}^{j}\neq \frac{\theta_{it}^{j}}{\alpha_{it}^{j}}.
    \]
    \item \textbf{Elasticity misspecification.}
    This implies that although
    \[
    \mu_{it}^{j}= \frac{\theta_{it}^{j}}{\alpha_{it}^{j}},
    \]
    the elasticity is not properly estimated:
    \[
    \theta_{it}^{j}\neq\hat{\theta}_{it}^{j}.
    \]
\end{enumerate}

\section{What can researchers do in the presence of falsification?}\label{SecFalsification}

As stated earlier, by sufficiently weakening the assumptions, a falsified baseline model becomes non-falsified. The continuous relaxations of our model assumptions are:

\begin{assumption}\label{AsFalsificationAdaptative}
Let $(\delta_{e,t}^{1}, \delta_{e,t}^{2}, \delta_{m,t}^{1}, \delta_{m,t}^{2})$ be a vector of nonnegative relaxation parameters.
\begin{enumerate}
    \item \textbf{Relaxation 1: bounded or frictional cost minimization:}
    \[
    \left|\mu_{it}^{j}-\frac{\theta_{it}^{j}}{\alpha_{it}^{j}}\right| \leq \delta_{m,t}^{j}.
    \]
    \item \textbf{Relaxation 2: imperfect elasticity identification:}
    \[
    |\theta_{it}^{j}-\hat{\theta}_{it}^{j}| \leq \delta_{e,t}^{j}.
    \]
\end{enumerate}
\end{assumption}

Relaxation 1 implies that equation \eqref{eq:markup} does not hold exactly. Instead, the true markup differs from the ratio of the output elasticity to the input share by at most \(\delta_{m,t}^{j}\). Relaxation 2 allows for imperfect identification of the production function. In particular, it implies that the estimation error of the output elasticity is bounded by \(\delta_{e,t}^{j}\).

Observe that, in practice, the \(\delta\)'s are relaxation parameters associated with the relevant assumptions. Although we allow them to depend on the particular variable input, in settings with no prior information on differential \(\delta\)'s it may be convenient to impose common upper bounds, denoted by \(\delta_{e,t}^{\max}\) and \(\delta_{m,t}^{\max}\), respectively.

It is worth noting that \((\delta_{e,t}^{1}, \delta_{e,t}^{2}, \delta_{m,t}^{1}, \delta_{m,t}^{2})=\mathbf{0}\) constitutes the null model from Section \ref{SecIdentification}, in which case the markup is point identified. On the other hand, \((\delta_{e,t}^{1}, \delta_{e,t}^{2}, \delta_{m,t}^{1}, \delta_{m,t}^{2})=\infty\) imposes no restrictions and therefore provides no information.

For the remainder of the paper, we use the notation
\[
\boldsymbol{\delta}_t=(\delta_{e,t}^{1}, \delta_{e,t}^{2}, \delta_{m,t}^{1}, \delta_{m,t}^{2}),
\qquad
\boldsymbol{\delta}_t^j=(\delta_{e,t}^{j}, \delta_{m,t}^{j}),
\qquad
\boldsymbol{\delta}_t^{\max}=(\delta_{e,t}^{\max}, \delta_{m,t}^{\max}),
\]
with
\[
\delta_{e,t}^{\max}=\max\{\delta_{e,t}^{1}, \delta_{e,t}^{2}\},
\qquad
\delta_{m,t}^{\max}=\max\{\delta_{m,t}^{1}, \delta_{m,t}^{2}\}.
\]

\paragraph{Parameter of interest.}
Several parameters may be relevant for researchers interested in markups, such as quantiles or values at mean inputs. In this paper we focus on the average markup at time $t$:
\[
\bar{\mu}_{t}\equiv \mathbb{E}[\mu_{it}].
\]

For any fixed level of $\boldsymbol{\delta}_t$ we have the following result.

\begin{proposition}\label{Proposition1}
Assume that prices and marginal costs are nonnegative. Under Assumption \ref{AsFalsificationAdaptative}, for a fixed level of $\boldsymbol{\delta}_t$, the average markup at time $t$ is partially identified as
\[
\max\{0, \max_{j=1,2}LB_{\boldsymbol{\delta}^{j}_{t}}\}
\;\leq\;
\bar{\mu}_{t}
\;\leq\;
\min_{j=1,2} UB_{\boldsymbol{\delta}^{j}_{t}},
\]
where
\begin{align*}
LB_{\boldsymbol{\delta}^{j}_{t}}
&\equiv
\mathbb{E}\Big(\frac{\hat{\theta}_{it}^{j}}{\alpha_{it}^{j}}\Big)
-
\delta_{e,t}^{j}\mathbb{E}\Big(\frac{1}{\alpha_{it}^{j}}\Big)
-
\delta_{m,t}^{j}, \\
UB_{\boldsymbol{\delta}^{j}_{t}}
&\equiv
\mathbb{E}\Big(\frac{\hat{\theta}_{it}^{j}}{\alpha_{it}^{j}}\Big)
+
\delta_{e,t}^{j}\mathbb{E}\Big(\frac{1}{\alpha_{it}^{j}}\Big)
+
\delta_{m,t}^{j}.
\end{align*}
\end{proposition}

The previous proposition extends the baseline analysis. When $\boldsymbol{\delta}_t=\mathbf{0}$ we recover the baseline point-identification result. The identified set becomes empty if the maximum lower bound exceeds the minimum upper bound.

In empirical work, where precise information on input-specific structural violations is rarely available, treating inputs symmetrically can be a useful conservative strategy. By adopting the uniform relaxation parameters $\delta_{e,t}^{\max}$ and $\delta_{m,t}^{\max}$, the following corollary provides a computationally convenient worst-case bound.

\begin{corollary}\label{Corollary1}
Under Assumption \ref{AsFalsificationAdaptative} with common relaxation levels
$\delta_{e,t}^{\max}$ and $\delta_{m,t}^{\max}$, the average markup at time $t$ is partially identified as
\[
\max\{0,\max_{j=1,2}\mathrm{LB}_j(\boldsymbol{\delta}_t^{\max})\}
\;\leq\;
\bar{\mu}_{t}
\;\leq\;
\min_{j=1,2} \mathrm{UB}_j(\boldsymbol{\delta}_t^{\max}),
\]
where
\begin{align*}
\mathrm{LB}_j(\boldsymbol{\delta}_t^{\max})
&\equiv
\mathbb{E}\Big(\frac{\hat{\theta}_{it}^{j}}{\alpha_{it}^{j}}\Big)
-
\delta_{e,t}^{\max}\mathbb{E}\Big(\frac{1}{\alpha_{it}^{j}}\Big)
-
\delta_{m,t}^{\max}, \\
\mathrm{UB}_j(\boldsymbol{\delta}_t^{\max})
&\equiv
\mathbb{E}\Big(\frac{\hat{\theta}_{it}^{j}}{\alpha_{it}^{j}}\Big)
+
\delta_{e,t}^{\max}\mathbb{E}\Big(\frac{1}{\alpha_{it}^{j}}\Big)
+
\delta_{m,t}^{\max}.
\end{align*}
\end{corollary}

It is important to emphasize that the identified set depends on the data through the moments
\[
\mathbb{E}\Big(\frac{1}{\alpha_{it}^{j}}\Big)
\qquad\text{and}\qquad
\mathbb{E}\Big(\frac{\hat{\theta}_{it}^{j}}{\alpha_{it}^{j}}\Big).
\]

Rather than forcing the researcher to choose specific values of $\boldsymbol{\delta}_t$, a more robust strategy is to ask: what is the minimal level of relaxation required to make the model compatible with the observed data? This leads to the falsification frontier.

\begin{proposition}\label{Proposition2}
Under Assumption \ref{AsFalsificationAdaptative} with common relaxation levels
$\delta_{e,t}^{\max}$ and $\delta_{m,t}^{\max}$, the falsification frontier is
\[
FF(\boldsymbol{\delta}_t)= FF_{1,2}(\boldsymbol{\delta}_t)\cup FF_{2,1}(\boldsymbol{\delta}_t),
\]
where
\[
FF_{j,j'}(\boldsymbol{\delta}_t)
\equiv
\left\{
(\delta_{m,t}^{\max}, \delta_{e,t}^{\max})
:
\delta_{m,t}^{\max}
=
\frac{
\mathbb{E}\left(\frac{\hat{\theta}_{it}^{j}}{\alpha_{it}^{j}}\right)
-
\mathbb{E}\left(\frac{\hat{\theta}_{it}^{j'}}{\alpha_{it}^{j'}}\right)
-
\delta_{e,t}^{\max}
\left[
\mathbb{E}\left(\frac{1}{\alpha_{it}^{j}}\right)
+
\mathbb{E}\left(\frac{1}{\alpha_{it}^{j'}}\right)
\right]
}{2}
\right\}.
\]
\end{proposition}

An immediate consequence is that the markup model is not falsifiable unless there are at least two variable inputs.

\begin{corollary}\label{Corollary2}
Under Assumption \ref{AsFalsificationAdaptative}, if there is only one variable input, the falsification frontier is empty.
\end{corollary}

The falsification adaptive set (FAS) is the union of all identified sets evaluated along the falsification frontier. When the baseline model is not falsified, the FAS collapses to the baseline singleton. When the baseline model is falsified, the FAS expands to reflect uncertainty about which minimally non-falsified relaxation profile is correct.

\begin{theorem}\label{Thm1}
Under Assumption \ref{AsFalsificationAdaptative} with common relaxation parameters $\delta_{e,t}^{\max}$ and $\delta_{m,t}^{\max}$, the falsification adaptive set for
\[
\bar{\mu}_{t} = \mathbb{E}[\mu_{it}]
\]
is determined by:
{\small
\begin{enumerate}
    \item
    \[
    FAS =
    \left(
    \mathbb{E}\left(\frac{\hat{\theta}_{it}^{2}}{\alpha_{it}^{2}}\right),
    \frac{
    \mathbb{E}\left(\frac{\hat{\theta}_{it}^{1}}{\alpha_{it}^{1}}\right)
    +
    \mathbb{E}\left(\frac{\hat{\theta}_{it}^{2}}{\alpha_{it}^{2}}\right)
    }{2}
    \right)
    \]
    if
    \[
    \mathbb{E}\left(\frac{1}{\alpha_{it}^{2}}\right)
    \leq
    \mathbb{E}\left(\frac{1}{\alpha_{it}^{1}}\right)
    \quad\text{and}\quad
    \mathbb{E}\left(\frac{\hat{\theta}_{it}^{1}}{\alpha_{it}^{1}}\right)
    \geq
    \mathbb{E}\left(\frac{\hat{\theta}_{it}^{2}}{\alpha_{it}^{2}}\right).
    \]

    \item
    \[
    FAS =
    \left(
    \frac{
    \mathbb{E}\left(\frac{\hat{\theta}_{it}^{1}}{\alpha_{it}^{1}}\right)
    +
    \mathbb{E}\left(\frac{\hat{\theta}_{it}^{2}}{\alpha_{it}^{2}}\right)
    }{2},
    \mathbb{E}\left(\frac{\hat{\theta}_{it}^{2}}{\alpha_{it}^{2}}\right)
    \right)
    \]
    if
    \[
    \mathbb{E}\left(\frac{1}{\alpha_{it}^{2}}\right)
    \leq
    \mathbb{E}\left(\frac{1}{\alpha_{it}^{1}}\right)
    \quad\text{and}\quad
    \mathbb{E}\left(\frac{\hat{\theta}_{it}^{1}}{\alpha_{it}^{1}}\right)
    <
    \mathbb{E}\left(\frac{\hat{\theta}_{it}^{2}}{\alpha_{it}^{2}}\right).
    \]

    \item
    \[
    FAS =
    \left(
    \frac{
    \mathbb{E}\left(\frac{\hat{\theta}_{it}^{1}}{\alpha_{it}^{1}}\right)
    +
    \mathbb{E}\left(\frac{\hat{\theta}_{it}^{2}}{\alpha_{it}^{2}}\right)
    }{2},
    \mathbb{E}\left(\frac{\hat{\theta}_{it}^{1}}{\alpha_{it}^{1}}\right)
    \right)
    \]
    if
    \[
    \mathbb{E}\left(\frac{1}{\alpha_{it}^{2}}\right)
    >
    \mathbb{E}\left(\frac{1}{\alpha_{it}^{1}}\right)
    \quad\text{and}\quad
    \mathbb{E}\left(\frac{\hat{\theta}_{it}^{1}}{\alpha_{it}^{1}}\right)
    \geq
    \mathbb{E}\left(\frac{\hat{\theta}_{it}^{2}}{\alpha_{it}^{2}}\right).
    \]

    \item
    \[
    FAS =
    \left(
    \mathbb{E}\left(\frac{\hat{\theta}_{it}^{1}}{\alpha_{it}^{1}}\right),
    \frac{
    \mathbb{E}\left(\frac{\hat{\theta}_{it}^{1}}{\alpha_{it}^{1}}\right)
    +
    \mathbb{E}\left(\frac{\hat{\theta}_{it}^{2}}{\alpha_{it}^{2}}\right)
    }{2}
    \right)
    \]
    if
    \[
    \mathbb{E}\left(\frac{1}{\alpha_{it}^{2}}\right)
    >
    \mathbb{E}\left(\frac{1}{\alpha_{it}^{1}}\right)
    \quad\text{and}\quad
    \mathbb{E}\left(\frac{\hat{\theta}_{it}^{1}}{\alpha_{it}^{1}}\right)
    <
    \mathbb{E}\left(\frac{\hat{\theta}_{it}^{2}}{\alpha_{it}^{2}}\right).
    \]
\end{enumerate}
}
\end{theorem}

Researchers may find it useful to present this set and its evolution over time alongside their baseline estimates, in order to assess the robustness of their conclusions and gain insight into the evolution of the average markup under weaker restrictions.

\subsection*{More than two variable inputs}
A comment is in place for settings with more than two variable inputs. Although it is possible to compute the FAS for $J$ variable inputs, the number of possible orderings of inputs and revenue-share moments quickly becomes unwieldy. We therefore suggest two alternative approaches, collected in Appendix \ref{Appmorethan2}.

\subsection*{Statistical Inference}

In finite samples, researchers can construct sample analog estimates of the falsification adaptive set together with corresponding confidence intervals. In our setting, statistical inference is complicated by the fact that the FAS endpoints are piecewise-defined functions of population moments. As a result, the global map is nonsmooth.

To formalize this, let
\[
s_t^j \equiv \mathbb{E}\!\left[\frac{1}{\alpha_{it}^{j}}\right],
\qquad
m_t^j \equiv \mathbb{E}\!\left[\frac{\hat{\theta}_{it}^{j}}{\alpha_{it}^{j}}\right],
\]
and let $\hat s_{t,n}^j$ and $\hat m_{t,n}^j$ denote their sample analogs. For a generic pair of variable inputs $(j,j')$, define
\[
\Delta s_t^{j,j'} \equiv s_t^{j'}-s_t^j,
\qquad
\Delta m_t^{j,j'} \equiv m_t^j-m_t^{j'}.
\]
The lower and upper FAS endpoints can be written as piecewise functions of these moments.

Although the global FAS map is nonsmooth, its differentiability properties are local. In particular, if
\[
\Delta s_t^{j,j'}\neq 0
\qquad\text{and}\qquad
\Delta m_t^{j,j'}\neq 0,
\]
then the true parameter lies strictly inside one regime of the piecewise formula rather than on a switching boundary. Consequently, there exists a neighborhood of the true parameter in which the active branch remains unchanged, so the endpoint map coincides locally with a single affine function. Therefore, under this regime-separation condition, each FAS endpoint is fully Hadamard differentiable at the true value.

This local differentiability result places our problem within the inferential framework of \cite{fang2019inference}. In particular, under a Gaussian limit law for the first-stage estimator, bootstrap consistency for the first-stage moments, and a nondegenerate scalar limit distribution, the bootstrap critical values constructed from the estimated directional derivative are consistent. This allows researchers to form asymptotically valid confidence intervals for each scalar endpoint of the FAS.

Appendix \ref{AppEInf} provides the formal differentiability argument, verifies Assumptions 2.1, 2.2, 3.1, 3.2 and 3.3 in \cite{fang2019inference}, and presents the construction of the corresponding bootstrap critical values and confidence intervals.

\section{Empirical Illustration}\label{SecIllustration}

Leveraging the replication package from \cite{raval2023testing}, we use the Chilean dataset employed in their analysis to conduct our own exercise. To construct the final dataset, \cite{raval2023testing} use plant-level manufacturing data for Chile covering the period 1979--1996. The source is the annual manufacturing census, Encuesta Nacional Industrial Anual (ENIA), which covers all Chilean manufacturing plants with at least 10 employees and contains information on roughly 5,000 plants per year.

The final dataset contains establishment-year information on capital, labor, materials, and sales. It also includes capital, materials, and output deflators, allowing for the construction of consistent input and output measures over time. Observations with zero or negative values for capital, labor, materials, sales, or labor costs are excluded from the sample. In addition, the data are cleaned by removing observations in the bottom and top 1\% of the distributions of labor’s revenue share, materials’ revenue share, and the composite variable-input revenue share within each industry. Labor is measured as the number of workers, while labor expenditures include salaries and fringe benefits. Materials expenditures include spending on raw materials, electricity, and fuels. The measure of capital is constructed as capital stocks multiplied by their rental rates, plus any rental payments for capital.

We estimate production functions and markups for each estimate of output elasticities using the \cite{ackerberg2015identification} (ACF) control-function estimator for a translog production function with capital, labor and materials, as in \cite{raval2023testing}.\footnote{For more details on the production function and control function, see \cite{raval2023testing}, Section 3.} Figure \ref{fig1} replicates Figure 1 from \cite{raval2023testing} for Chilean food products and illustrates the dispersion in markup estimates across inputs.

\begin{figure}[h]
\caption{Distribution of Translog Markups for Chilean Food Products}
\label{fig1}
\centering
\includegraphics[width=1\linewidth]{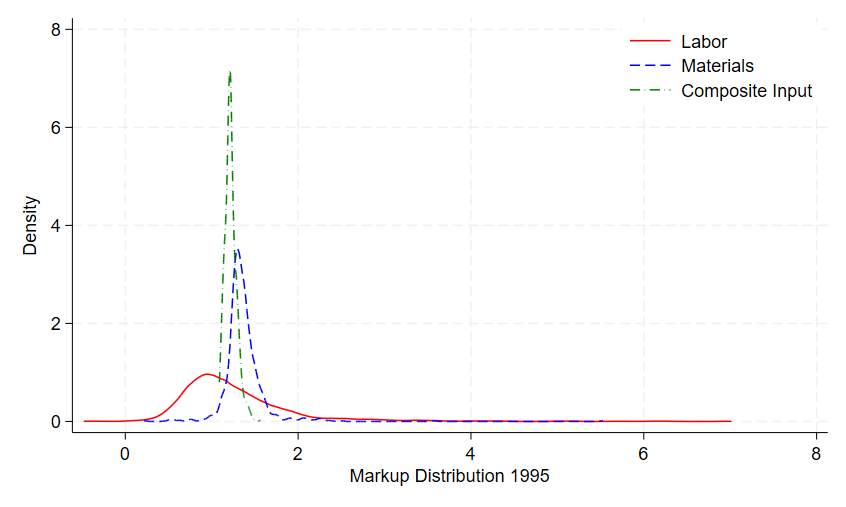}
\end{figure}

The dispersion observed in Figure \ref{fig1} translates into different average markup estimates across inputs, as shown in Figure \ref{fig2}. Not only the levels differ, but also the implied time trend. While one series presents an inverted-U shape, the other suggests a steady increase in markups. Thus, conclusions are highly sensitive to the selected input.

\begin{figure}[h]
\caption{Average Estimate of the Markup Over Time}
\label{fig2}
\centering
\includegraphics[width=1\linewidth]{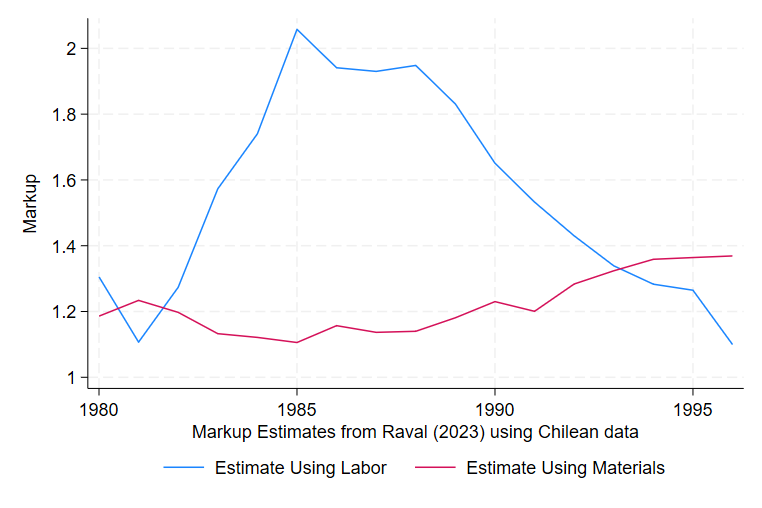}
\end{figure}

To obtain a robust measure of industry markups under potential misspecification, we compute the FAS, which is illustrated in Figure~\ref{fig:MainFAS}. Once model relaxations are incorporated, the sharp conclusions regarding time trends disappear. Specifically, the substantial overlap between the identified sets across years precludes any definitive statement about the directional evolution of markups. For instance, the FAS for 1985 includes parameter values that are simultaneously larger and smaller than those contained in the 1995 set, rendering traditional trend analysis non-conclusive.

In this sense, our framework serves as a valuable diagnostic tool to evaluate the robustness of time-series interpretations commonly found in the markup literature. This ambiguity is entirely consistent with the divergent input-specific trends previously documented in Figure~\ref{fig2}. Despite this loss of directional certainty, the FAS provides one highly robust conclusion: across the entire sample period, the lower bound of the identified set remains strictly above one. Consequently, even after allowing for structural violations of the baseline model, we can robustly reject the hypothesis of perfect competition.

\begin{figure}[h]
\caption{Estimated FAS}
\label{fig:MainFAS}
\centering
\includegraphics[width=1\linewidth]{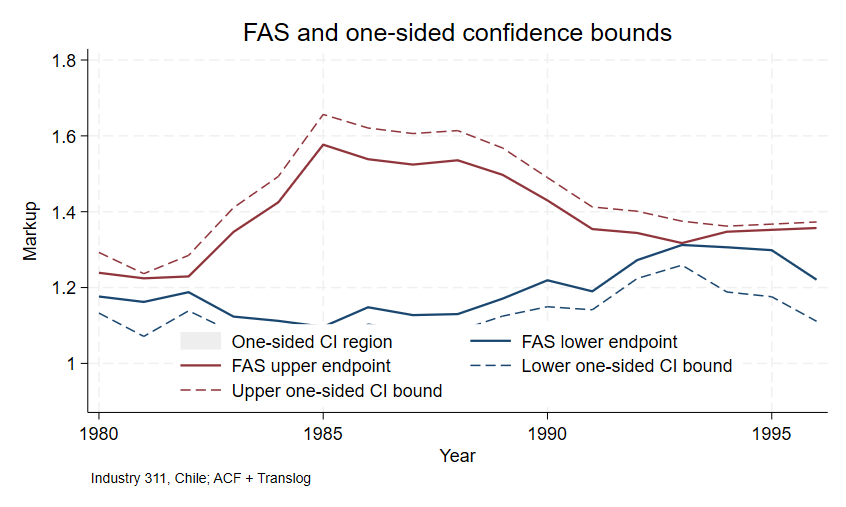}
\begin{flushleft}
\small \textit{Note:} One sided 95 percent confidence intervals, estimated from 500 Bootstrap replications following \cite{fang2019inference}.  
\end{flushleft}
\end{figure}
\section{Conclusions}\label{SecConclusions}

In this paper we provide a constructive way for researchers to salvage the \cite{loecker2012markups} markup model when falsified. To do this, we consider continuous relaxations of the assumptions underlying static cost minimization of variable inputs and correct estimation of the input-output elasticity. By computing the values of the markup as a function of the relaxations across the minimal set of non-falsified models, we obtain an identified set for the markup which generalizes the standard baseline markup estimand to account for possible falsification without the need to impose additional or new assumptions on the production function or firm behavior. We illustrate this using data from \cite{raval2023testing} and conclude that there exists robust evidence that markups exceed one, but that interpretations about the shape of the markup trend over time should be treated with caution, since they are not robust once misspecification is acknowledged.

\renewcommand{\refname}{References}
\bibliographystyle{chicago}
\bibliography{references}

\appendix

\section{Proofs}\label{Appproofs}

\subsection{Proof of the main results}

\begin{proof}[Proof of Proposition \ref{Proposition1}]
Recall that under Assumption \ref{AsFalsificationAdaptative}, part 1,
\begin{align*}
\frac{\theta_{it}^{1}}{\alpha_{it}^{1}}- \delta_{m,t}^{1}
&\leq \mu_{it} \leq
\frac{\theta_{it}^{1}}{\alpha_{it}^{1}}+ \delta_{m,t}^{1}, \\
\frac{\theta_{it}^{2}}{\alpha_{it}^{2}}- \delta_{m,t}^{2}
&\leq \mu_{it} \leq
\frac{\theta_{it}^{2}}{\alpha_{it}^{2}}+ \delta_{m,t}^{2}.
\end{align*}
Now, by Assumption \ref{AsFalsificationAdaptative}, part 2,
\begin{align*}
\frac{\hat{\theta}_{it}^{1}}{\alpha_{it}^{1}}
- \delta_{e,t}^{1}\frac{1}{\alpha_{it}^{1}}
- \delta_{m,t}^{1}
&\leq \mu_{it} \leq
\frac{\hat{\theta}_{it}^{1}}{\alpha_{it}^{1}}
+\delta_{e,t}^{1}\frac{1}{\alpha_{it}^{1}}
+ \delta_{m,t}^{1}, \\
\frac{\hat{\theta}_{it}^{2}}{\alpha_{it}^{2}}
- \delta_{e,t}^{2}\frac{1}{\alpha_{it}^{2}}
- \delta_{m,t}^{2}
&\leq \mu_{it} \leq
\frac{\hat{\theta}_{it}^{2}}{\alpha_{it}^{2}}
+\delta_{e,t}^{2}\frac{1}{\alpha_{it}^{2}}
+ \delta_{m,t}^{2}.
\end{align*}
We furthermore know that markups must be nonnegative:
\[
0\leq \mu_{it}\leq +\infty.
\]
Taking expectations yields the corresponding bounds for $\bar{\mu}_t$. Intersecting the resulting intervals gives the claim.
\end{proof}

\begin{proof}[Proof of Proposition \ref{Proposition2}]
It should be noted that the identified set of the average markup is an interval. Thus, the set would be empty if $\max\{0, \max_{j=1,2}\text{LB}_j(\boldsymbol{\delta}_t^{\max})\}>  \min_{j=1,2}\text{UB}_j(\boldsymbol{\delta}_t^{\max})$
This can happen in one if several scenarios depending on which element is the maximum and which one is the minimum. We now proceed on a case by case basis.
\par 
\textbf{Case 1:} 
$$\max\{0, \max_{j=1,2}\text{LB}_j(\boldsymbol{\delta}_t^{\max})\}=\mathbb{E}\Big(\frac{\hat{\theta}_{it}^{1}}{ \alpha_{it}^{1}}\Big)- \delta_{e,t}^{\text{max}}\mathbb{E}\Big(\frac{1}{ \alpha_{it}^{1}}\Big)- \delta_{m,t}^{\text{max}}$$
$$ \min_{j=1,2}\text{UB}_j(\boldsymbol{\delta}_t^{\max})=\mathbb{E}\Big(\frac{\hat{\theta}_{it}^{1}}{ \alpha_{it}^{1}}\Big)+\delta_{e,t}^{\text{max}}\mathbb{E}\Big(\frac{1}{ \alpha_{it}^{1}}\Big)+ \delta_{m,t}^{\text{max}}$$
It should be noted that this happens when 
\begin{eqnarray*}
    \mathbb{E}\Big(\frac{\hat{\theta}_{it}^{1}}{ \alpha_{it}^{1}}\Big)- \delta_{e,t}^{\text{max}}\mathbb{E}\Big(\frac{1}{ \alpha_{it}^{1}}\Big)- \delta_{m,t}^{\text{max}}&\geq& 0 \\
    \mathbb{E}\Big(\frac{\hat{\theta}_{it}^{1}}{ \alpha_{it}^{1}}\Big)- \delta_{e,t}^{\text{max}}\mathbb{E}\Big(\frac{1}{ \alpha_{it}^{1}}\Big)- \delta_{m,t}^{\text{max}}&\geq& \mathbb{E}\Big(\frac{\hat{\theta}_{it}^{2}}{ \alpha_{it}^{2}}\Big)- \delta_{e,t}^{\text{max}}\mathbb{E}\Big(\frac{1}{ \alpha_{it}^{2}}\Big)- \delta_{m,t}^{\text{max}}\\
    \mathbb{E}\Big(\frac{\hat{\theta}_{it}^{1}}{ \alpha_{it}^{1}}\Big)+\delta_{e,t}^{\text{max}}\mathbb{E}\Big(\frac{1}{ \alpha_{it}^{1}}\Big)+ \delta_{m,t}^{\text{max}}&\leq& \mathbb{E}\Big(\frac{\hat{\theta}_{it}^{2}}{ \alpha_{it}^{2}}\Big)+\delta_{e,t}^{\text{max}}\mathbb{E}\Big(\frac{1}{ \alpha_{it}^{2}}\Big)+ \delta_{m,t}^{\text{max}}
\end{eqnarray*}
Or more compactly: 
\begin{eqnarray*}
    \mathbb{E}\Big(\frac{\hat{\theta}_{it}^{1}}{ \alpha_{it}^{1}}\Big)- \delta_{e,t}^{\text{max}}\mathbb{E}\Big(\frac{1}{ \alpha_{it}^{1}}\Big) &\geq& \delta_{m,t}^{\text{max}}\\
    \mathbb{E}\Big(\frac{\hat{\theta}_{it}^{1}}{ \alpha_{it}^{1}}\Big)-\mathbb{E}\Big(\frac{\hat{\theta}_{it}^{2}}{ \alpha_{it}^{2}}\Big)+ \delta_{e,t}^{\text{max}}\Bigg( \mathbb{E}\Big(\frac{1}{ \alpha_{it}^{2}}\Big)-\mathbb{E}\Big(\frac{1}{ \alpha_{it}^{1}}\Big|)\Bigg) &\geq& 0 \\
    \mathbb{E}\Big(\frac{\hat{\theta}_{it}^{1}}{ \alpha_{it}^{1}}\Big)-\mathbb{E}\Big(\frac{\hat{\theta}_{it}^{2}}{ \alpha_{it}^{2}}\Big)-\delta_{e,t}^{\text{max}}\Bigg(\mathbb{E}\Big(\frac{1}{ \alpha_{it}^{2}}\Big)-\mathbb{E}\Big(\frac{1}{ \alpha_{it}^{1}}\Big) \Bigg) &\leq& 0
\end{eqnarray*}
The condition that would then be consistent with an empty identified set would be: 
\begin{eqnarray*}
\mathbb{E}\Big(\frac{\hat{\theta}_{it}^{1}}{ \alpha_{it}^{1}}\Big)- \delta_{e,t}^{\text{max}}\mathbb{E}\Big(\frac{1}{ \alpha_{it}^{1}}\Big)- \delta_{m,t}^{\text{max}}> \mathbb{E}\Big(\frac{\hat{\theta}_{it}^{1}}{ \alpha_{it}^{1}}\Big)+\delta_{e,t}^{\text{max}}\mathbb{E}\Big(\frac{1}{ \alpha_{it}^{1}}\Big)+ \delta_{m,t}^{\text{max}}  
\end{eqnarray*}
From operating with the previous we get: 
\begin{eqnarray*}
\delta_{e,t}^{\text{max}}\mathbb{E}\Big(\frac{1}{ \alpha_{it}^{1}}\Big) < - \delta_{m,t}^{\text{max}}  
\end{eqnarray*}
But note that $\delta_{e,t}^{\text{max}}\geq 0, -\delta_{m,t}^{\text{max}}\leq 0, \mathbb{E}\Big(\frac{1}{ \alpha_{it}^{1}}\Big)> 0$ thus it can never happen that $\delta_{e,t}^{\text{max}}\mathbb{E}\Big(\frac{1}{ \alpha_{it}^{1}}\Big) < - \delta_{m,t}^{\text{max}} $
Thus, no rejection is possible in this scenario. Thus $FF=\emptyset$ in this case.  
\par 
\textbf{Case 2:} 
$$\max\{0, \max_{j=1,2}\text{LB}_j(\boldsymbol{\delta}_t^{\max})\}=\mathbb{E}\Big(\frac{\hat{\theta}_{it}^{2}}{ \alpha_{it}^{2}}\Big)- \delta_{e,t}^{\text{max}}\mathbb{E}\Big(\frac{1}{ \alpha_{it}^{2}}\Big)- \delta_{m,t}^{\text{max}}$$
$$ \min_{j=1,2}\text{UB}_j(\boldsymbol{\delta}_t^{\max})=\mathbb{E}\Big(\frac{\hat{\theta}_{it}^{2}}{ \alpha_{it}^{2}}\Big)+\delta_{e,t}^{\text{max}}\mathbb{E}\Big(\frac{1}{ \alpha_{it}^{2}}\Big)+ \delta_{m,t}^{\text{max}}$$
It should be noted that this happens when: 
\begin{eqnarray*}
    \mathbb{E}\Big(\frac{\hat{\theta}_{it}^{2}}{ \alpha_{it}^{2}}\Big)- \delta_{e,t}^{\text{max}}\mathbb{E}\Big(\frac{1}{ \alpha_{it}^{2}}\Big) &\geq& \delta_{m,t}^{\text{max}}\\
    \mathbb{E}\Big(\frac{\hat{\theta}_{it}^{1}}{ \alpha_{it}^{1}}\Big)-\mathbb{E}\Big(\frac{\hat{\theta}_{it}^{2}}{ \alpha_{it}^{2}}\Big)+ \delta_{e,t}^{\text{max}}\Bigg( \mathbb{E}\Big(\frac{1}{ \alpha_{it}^{2}}\Big)-\mathbb{E}\Big(\frac{1}{ \alpha_{it}^{1}}\Big)\Bigg) &\leq& 0 \\
    \mathbb{E}\Big(\frac{\hat{\theta}_{it}^{1}}{ \alpha_{it}^{1}}\Big)-\mathbb{E}\Big(\frac{\hat{\theta}_{it}^{2}}{ \alpha_{it}^{2}}\Big)-\delta_{e,t}^{\text{max}}\Bigg(\mathbb{E}\Big(\frac{1}{ \alpha_{it}^{2}}\Big)-\mathbb{E}\Big(\frac{1}{ \alpha_{it}^{1}}\Big) \Bigg) &\geq& 0
\end{eqnarray*}
By a similar display as in Case 1 we would never get an empty-set in this scenario. Thus $FF=\emptyset$ in this case too.  
\par 
\textbf{Case 3:}
$$\max\{0, \max_{j=1,2}\text{LB}_j(\boldsymbol{\delta}_t^{\max})\}=0$$
$$ \min_{j=1,2}\text{UB}_j(\boldsymbol{\delta}_t^{\max})=\mathbb{E}\Big(\frac{\hat{\theta}_{it}^{1}}{ \alpha_{it}^{1}}\Big)+\delta_{e,t}^{\text{max}}\mathbb{E}\Big(\frac{1}{ \alpha_{it}^{1}}\Big)+ \delta_{m,t}^{\text{max}}$$
This happens when:
\begin{eqnarray*}
    \mathbb{E}\Big(\frac{\hat{\theta}_{it}^{1}}{ \alpha_{it}^{1}}\Big)- \delta_{e,t}^{\text{max}}\mathbb{E}\Big(\frac{1}{ \alpha_{it}^{1}}\Big) &\leq& \delta_{m,t}^{\text{max}}\\
\mathbb{E}\Big(\frac{\hat{\theta}_{it}^{2}}{ \alpha_{it}^{2}}\Big)- \delta_{e,t}^{\text{max}}\mathbb{E}\Big(\frac{1}{ \alpha_{it}^{2}}\Big) &\leq& \delta_{m,t}^{\text{max}}\\
    \mathbb{E}\Big(\frac{\hat{\theta}_{it}^{1}}{ \alpha_{it}^{1}}\Big)-\mathbb{E}\Big(\frac{\hat{\theta}_{it}^{2}}{ \alpha_{it}^{2}}\Big)-\delta_{e,t}^{\text{max}}\Bigg(\mathbb{E}\Big(\frac{1}{ \alpha_{it}^{2}}\Big)-\mathbb{E}\Big(\frac{1}{ \alpha_{it}^{1}}\Big) \Bigg) &\leq& 0
\end{eqnarray*}
The condition that would then be consistent with an empty  identified set would be: 
\begin{eqnarray*}
    0 &>& \mathbb{E}\Big(\frac{\hat{\theta}_{it}^{1}}{ \alpha_{it}^{1}}\Big)+\delta_{e,t}^{\text{max}}\mathbb{E}\Big(\frac{1}{ \alpha_{it}^{1}}\Big)+ \delta_{m,t}^{\text{max}}  \\
    &\iff& \\
    \delta_{m,t}^{\text{max}} &<& -\mathbb{E}\Big(\frac{\hat{\theta}_{it}^{1}}{ \alpha_{it}^{1}}\Big)-\delta_{e,t}^{\text{max}}\mathbb{E}\Big(\frac{1}{ \alpha_{it}^{1}}\Big)
\end{eqnarray*}
Importantly $\delta_{m,t}^{\text{max}}, \theta_t^{1}, \mathbb{E}\Big(\frac{1}{ \alpha_{it}^{1}}\Big), \delta_{e,t}^{\text{max}}, \mathbb{E}\Big(\frac{1}{ \alpha_{it}^{1}}\Big)$ are all greater or equal to $0$ thus the previous can never happen. Then $FF=\emptyset$ in this case. 
\par 
\textbf{Case 4:}
$$\max\{0, \max_{j=1,2}\text{LB}_j(\boldsymbol{\delta}_t^{\max})\}=0$$
$$ \min_{j=1,2}\text{UB}_j(\boldsymbol{\delta}_t^{\max})=\mathbb{E}\Big(\frac{\hat{\theta}_{it}^{2}}{ \alpha_{it}^{2}}\Big)+\delta_{e,t}^{\text{max}}\mathbb{E}\Big(\frac{1}{ \alpha_{it}^{2}}\Big)+ \delta_{m,t}^{\text{max}}$$
This happens when:
\begin{eqnarray*}
    \mathbb{E}\Big(\frac{\hat{\theta}_{it}^{1}}{ \alpha_{it}^{1}}\Big)- \delta_{e,t}^{\text{max}}\mathbb{E}\Big(\frac{1}{ \alpha_{it}^{1}}\Big) &\leq& \delta_{m,t}^{\text{max}}\\
\mathbb{E}\Big(\frac{\hat{\theta}_{it}^{2}}{ \alpha_{it}^{2}}\Big)- \delta_{e,t}^{\text{max}}\mathbb{E}\Big(\frac{1}{ \alpha_{it}^{2}}\Big) &\leq& \delta_{m,t}^{\text{max}}\\
    \mathbb{E}\Big(\frac{\hat{\theta}_{it}^{1}}{ \alpha_{it}^{1}}\Big)-\mathbb{E}\Big(\frac{\hat{\theta}_{it}^{2}}{ \alpha_{it}^{2}}\Big)-\delta_{e,t}^{\text{max}}\Bigg(\mathbb{E}\Big(\frac{1}{ \alpha_{it}^{2}}\Big)-\mathbb{E}\Big(\frac{1}{ \alpha_{it}^{1}}\Big) \Bigg) &\geq& 0
\end{eqnarray*}
Similar to Case 3, $FF=\emptyset$. 
\par
\textbf{Case 5:}
$$\max\{0, \max_{j=1,2}\text{LB}_j(\boldsymbol{\delta}_t^{\max})\}=\mathbb{E}\Big(\frac{\hat{\theta}_{it}^{1}}{ \alpha_{it}^{1}}\Big)- \delta_{e,t}^{\text{max}}\mathbb{E}\Big(\frac{1}{ \alpha_{it}^{1}}\Big)- \delta_{m,t}^{\text{max}}$$
$$ \min_{j=1,2}\text{UB}_j(\boldsymbol{\delta}_t^{\max})=\mathbb{E}\Big(\frac{\hat{\theta}_{it}^{2}}{ \alpha_{it}^{2}}\Big)+\delta_{e,t}^{\text{max}}\mathbb{E}\Big(\frac{1}{ \alpha_{it}^{2}}\Big)+ \delta_{m,t}^{\text{max}}$$
Observe that this happens when 
\begin{eqnarray*}
    \mathbb{E}\Big(\frac{\hat{\theta}_{it}^{1}}{ \alpha_{it}^{1}}\Big)- \delta_{e,t}^{\text{max}}\mathbb{E}\Big(\frac{1}{ \alpha_{it}^{1}}\Big) &\geq& \delta_{m,t}^{\text{max}}\\
    \mathbb{E}\Big(\frac{\hat{\theta}_{it}^{1}}{ \alpha_{it}^{1}}\Big)-\mathbb{E}\Big(\frac{\hat{\theta}_{it}^{2}}{ \alpha_{it}^{2}}\Big)+ \delta_{e,t}^{\text{max}}\Bigg( \mathbb{E}\Big(\frac{1}{ \alpha_{it}^{2}}\Big)-\mathbb{E}\Big(\frac{1}{ \alpha_{it}^{1}}\Big)\Bigg) &\geq& 0 \\
    \mathbb{E}\Big(\frac{\hat{\theta}_{it}^{1}}{ \alpha_{it}^{1}}\Big)-\mathbb{E}\Big(\frac{\hat{\theta}_{it}^{2}}{ \alpha_{it}^{2}}\Big)-\delta_{e,t}^{\text{max}}\Bigg(\mathbb{E}\Big(\frac{1}{ \alpha_{it}^{2}}\Big)-\mathbb{E}\Big(\frac{1}{ \alpha_{it}^{1}}\Big) \Bigg) &\geq& 0
\end{eqnarray*}
The condition that would then be consistent with an empty identified set would be: 
\begin{eqnarray*}
\mathbb{E}\Big(\frac{\hat{\theta}_{it}^{1}}{ \alpha_{it}^{1}}\Big)- \delta_{e,t}^{\text{max}}\mathbb{E}\Big(\frac{1}{ \alpha_{it}^{1}}\Big)- \delta_{m,t}^{\text{max}}&>& \mathbb{E}\Big(\frac{\hat{\theta}_{it}^{2}}{ \alpha_{it}^{2}}\Big)+\delta_{e,t}^{\text{max}}\mathbb{E}\Big(\frac{1}{ \alpha_{it}^{2}}\Big)+ \delta_{m,t}^{\text{max}}
\end{eqnarray*}
Or equivalently: 
\begin{eqnarray*}
    \delta_{m,t}^{\text{max}}  &<& \frac{\mathbb{E}\Big(\frac{\hat{\theta}_{it}^{1}}{ \alpha_{it}^{1}}\Big)-\mathbb{E}\Big(\frac{\hat{\theta}_{it}^{2}}{ \alpha_{it}^{2}}\Big)-  \delta_{e,t}^{\text{max}} \Bigg(\mathbb{E}\Big(\frac{1}{ \alpha_{it}^{1}}\Big)+ \mathbb{E}\Big(\frac{1}{ \alpha_{it}^{2}}\Big) \Bigg)}{2}
\end{eqnarray*}

This implies that the falsification frontier is given by the first non-falsified model:

\begin{eqnarray*}
 FF=\Bigg\{ \delta_{m,t}^{\text{max}}, \delta_{e,t}^{\text{max}} \in \mathbb{R}_+^2: \delta_{m,t}^{\text{max}}  = \frac{\mathbb{E}\Big(\frac{\hat{\theta}_{it}^{1}}{ \alpha_{it}^{1}}\Big)-\mathbb{E}\Big(\frac{\hat{\theta}_{it}^{2}}{ \alpha_{it}^{2}}\Big)-  \delta_{e,t}^{\text{max}} \Bigg(\mathbb{E}\Big(\frac{1}{ \alpha_{it}^{1}}\Big)+ \mathbb{E}\Big(\frac{1}{ \alpha_{it}^{2}}\Big) \Bigg)}{2} \Bigg\}   
\end{eqnarray*}
\par
\textbf{Case 6:}
$$\max\{0, \max_{j=1,2}\text{LB}_j(\boldsymbol{\delta}_t^{\max})\}=\mathbb{E}\Big(\frac{\hat{\theta}_{it}^{2}}{ \alpha_{it}^{2}}\Big)- \delta_{e,t}^{\text{max}}\mathbb{E}\Big(\frac{1}{ \alpha_{it}^{2}}\Big)- \delta_{m,t}^{\text{max}}$$
$$ \min_{j=1,2}\text{UB}_j(\boldsymbol{\delta}_t^{\max})=\mathbb{E}\Big(\frac{\hat{\theta}_{it}^{1}}{ \alpha_{it}^{1}}\Big)+\delta_{e,t}^{\text{max}}\mathbb{E}\Big(\frac{1}{ \alpha_{it}^{1}}\Big)+ \delta_{m,t}^{\text{max}}$$
Observe that this happens when 
\begin{eqnarray*}
    \mathbb{E}\Big(\frac{\hat{\theta}_{it}^{2}}{ \alpha_{it}^{2}}\Big)- \delta_{e,t}^{\text{max}}\mathbb{E}\Big(\frac{1}{ \alpha_{it}^{2}}\Big) &\geq& \delta_{m,t}^{\text{max}}\\
    \mathbb{E}\Big(\frac{\hat{\theta}_{it}^{1}}{ \alpha_{it}^{1}}\Big)-\mathbb{E}\Big(\frac{\hat{\theta}_{it}^{2}}{ \alpha_{it}^{2}}\Big)+ \delta_{e,t}^{\text{max}}\Bigg( \mathbb{E}\Big(\frac{1}{ \alpha_{it}^{2}}\Big)-\mathbb{E}\Big(\frac{1}{ \alpha_{it}^{1}}\Big)\Bigg) &\leq& 0 \\
    \mathbb{E}\Big(\frac{\hat{\theta}_{it}^{1}}{ \alpha_{it}^{1}}\Big)-\mathbb{E}\Big(\frac{\hat{\theta}_{it}^{2}}{ \alpha_{it}^{2}}\Big)-\delta_{e,t}^{\text{max}}\Bigg(\mathbb{E}\Big(\frac{1}{ \alpha_{it}^{2}}\Big)-\mathbb{E}\Big(\frac{1}{ \alpha_{it}^{1}}\Big) \Bigg) &\leq& 0
\end{eqnarray*}
Then by a similar logic as in Case 5: 
\begin{eqnarray*}
 FF=\Bigg\{ \delta_{m,t}^{\text{max}}, \delta_{e,t}^{\text{max}} \in \mathbb{R}_+^2: \delta_{m,t}^{\text{max}}  = \frac{\mathbb{E}\Big(\frac{\hat{\theta}_{it}^{2}}{ \alpha_{it}^{2}}\Big)-\mathbb{E}\Big(\frac{\hat{\theta}_{it}^{1}}{ \alpha_{it}^{1}}\Big)-  \delta_{e,t}^{\text{max}} \Bigg(\mathbb{E}\Big(\frac{1}{ \alpha_{it}^{2}}\Big)+ \mathbb{E}\Big(\frac{1}{ \alpha_{it}^{1}}\Big) \Bigg)}{2} \Bigg\}   
\end{eqnarray*}

\cite{masten2021salvaging} define the falsification frontier as the set of assumptions that are not falsified for a given $\boldsymbol{\delta}_t$, but whose strengthening in any component yields a falsified model. Because Cases~1--4 exhibit no falsification and thus lack a falsification frontier, the overall frontier equals the union of the frontiers from the final two cases, which do possess falsification conditions.

\end{proof}

\begin{proof}[Proof of Corollary \ref{Corollary2}]
 With only one input the average markup can be bounded by:
\begin{eqnarray*}
\Bigg(\max\{0,\mathbb{E}\Big(\frac{\hat{\theta}_{it}^{1}}{ \alpha_{it}^{1}}\Big)- \delta_{e,t}^{\text{max}}\mathbb{E}\Big(\frac{1}{ \alpha_{it}^{1}}\Big)- \delta_{m,t}^{\text{max}}\},\mathbb{E}\Big(\frac{\hat{\theta}_{it}^{1}}{ \alpha_{it}^{1}}\Big)+\delta_{e,t}^{\text{max}}\mathbb{E}\Big(\frac{1}{ \alpha_{it}^{1}}\Big)+ \delta_{m,t}^{\text{max}}\Bigg)
\end{eqnarray*}
 This is an application of Proposition \ref{Proposition1} with only one variable input. 
\par 
By a similar logic as in Proposition \ref{Proposition2} we have  cases. 
\par 
\begin{enumerate}
    \item \textbf{When } $\mathbb{E}\Big(\frac{\hat{\theta}_{it}^{1}}{ \alpha_{it}^{1}}\Big)- \delta_{e,t}^{\text{max}}\mathbb{E}\Big(\frac{1}{ \alpha_{it}^{1}}\Big)- \delta_{m,t}^{\text{max}}>0$
In this case the bounds look like: 

\begin{eqnarray*}
\Bigg(\mathbb{E}\Big(\frac{\hat{\theta}_{it}^{1}}{ \alpha_{it}^{1}}\Big)- \delta_{e,t}^{\text{max}}\mathbb{E}\Big(\frac{1}{ \alpha_{it}^{1}}\Big)- \delta_{m,t}^{\text{max}},\mathbb{E}\Big(\frac{\hat{\theta}_{it}^{1}}{ \alpha_{it}^{1}}\Big)+\delta_{e,t}^{\text{max}}\mathbb{E}\Big(\frac{1}{ \alpha_{it}^{1}}\Big)+ \delta_{m,t}^{\text{max}}\Bigg)
\end{eqnarray*}
Thus the model would be faslified if: 
\begin{eqnarray*}
2(- \delta_{e,t}^{\text{max}}\mathbb{E}\Big(\frac{1}{ \alpha_{it}^{1}}\Big)- \delta_{m,t}^{\text{max}}) >0 
\end{eqnarray*}
Which never happens since  $\delta_{e,t}^{\text{max}},\mathbb{E}\Big(\frac{1}{ \alpha_{it}^{1}}\Big), \delta_{m,t}^{\text{max}}$ are nonnegative.  Thus no falsification would be possible in this case.
    \item \textbf{When} $\mathbb{E}\Big(\frac{\hat{\theta}_{it}^{1}}{ \alpha_{it}^{1}}\Big)- \delta_{e,t}^{\text{max}}\mathbb{E}\Big(\frac{1}{ \alpha_{it}^{1}}\Big)- \delta_{m,t}^{\text{max}}\leq 0$
In this case the bounds look like: 
\begin{eqnarray*}
\Bigg(0,\mathbb{E}\Big(\frac{\hat{\theta}_{it}^{1}}{ \alpha_{it}^{1}}\Big)+\delta_{e,t}^{\text{max}}\mathbb{E}\Big(\frac{1}{ \alpha_{it}^{1}}\Big)+ \delta_{m,t}^{\text{max}}\Bigg)
\end{eqnarray*}
Thus the model would be faslified if: 
$$\mathbb{E}\Big(\frac{\hat{\theta}_{it}^{1}}{ \alpha_{it}^{1}}\Big)+\delta_{e,t}^{\text{max}}\mathbb{E}\Big(\frac{1}{ \alpha_{it}^{1}}\Big)+ \delta_{m,t}^{\text{max}}<0  $$
Which never happens since $\mathbb{E}\Big(\frac{\hat{\theta}_{it}^{1}}{ \alpha_{it}^{1}}\Big),\delta_{e,t}^{\text{max}}, \mathbb{E}\Big(\frac{1}{ \alpha_{it}^{1}}\Big), \delta_{m,t}^{\text{max}}$ are all positive numbers. 
Thus no falsification would be possible in this case either
\end{enumerate} 
\end{proof}

\begin{proof}[Proof of Theorem \ref{Thm1}]
The FAS will be the union across values of $\boldsymbol{\delta}_t$ in the FF. Since we have different bounds that imply different frontiers we will compute the  FAS associated with each of the possible bounds and their associated FF. Since we have empty falsification frontiers for several possible bound(some of them are not falsifiable) except for two of them the FAS will be the union across the $\boldsymbol{\delta}_t$ consistent with the two bounds ( the ones that are indeed falsifiable).

 Observe that for the case that: \begin{eqnarray}\label{Eq0Thm1}
    \mathbb{E}\Big(\frac{\hat{\theta}_{it}^{1}}{ \alpha_{it}^{1}}\Big)- \delta_{e,t}^{\text{max}}\mathbb{E}\Big(\frac{1}{ \alpha_{it}^{1}}\Big) &\geq& \delta_{m,t}^{\text{max}}\nonumber \\
    \mathbb{E}\Big(\frac{\hat{\theta}_{it}^{1}}{ \alpha_{it}^{1}}\Big)-\mathbb{E}\Big(\frac{\hat{\theta}_{it}^{2}}{ \alpha_{it}^{2}}\Big)+ \delta_{e,t}^{\text{max}}\Bigg( \mathbb{E}\Big(\frac{1}{ \alpha_{it}^{2}}\Big)-\mathbb{E}\Big(\frac{1}{ \alpha_{it}^{1}}\Big)\Bigg) &\geq& 0 \nonumber \\
    \mathbb{E}\Big(\frac{\hat{\theta}_{it}^{1}}{ \alpha_{it}^{1}}\Big)-\mathbb{E}\Big(\frac{\hat{\theta}_{it}^{2}}{ \alpha_{it}^{2}}\Big)-\delta_{e,t}^{\text{max}}\Bigg(\mathbb{E}\Big(\frac{1}{ \alpha_{it}^{2}}\Big)-\mathbb{E}\Big(\frac{1}{ \alpha_{it}^{1}}\Big) \Bigg) &\geq& 0 \nonumber \\
\end{eqnarray}
For an element of $FF(\boldsymbol{\delta}_t)$ say $\delta_{m,t}^{\text{max}, FF}, \delta_{e,t}^{\text{max}, FF}$ we have: 
\begin{eqnarray}\label{Eq1Thm1}
\max\{0,\max_{j=1,2}LB_{\delta^{FF},m,t}\}&=&\min_{j=1,2} = UB_{\delta^{FF},m,t}    \nonumber \\
&=& \mathbb{E}\Big(\frac{\hat{\theta}_{it}^{2}}{ \alpha_{it}^{2}}\Big)+\delta_{e,t}^{\text{max}, FF}\mathbb{E}\Big(\frac{1}{ \alpha_{it}^{2}}\Big)+ \delta_{m,t}^{\text{max}, FF} \nonumber \\
&=& \mathbb{E}\Big(\frac{\hat{\theta}_{it}^{2}}{ \alpha_{it}^{2}}\Big)+\delta_{e,t}^{\text{max}, FF}\mathbb{E}\Big(\frac{1}{ \alpha_{it}^{2}}\Big) \nonumber \\
&+& \frac{\mathbb{E}\Big(\frac{\hat{\theta}_{it}^{1}}{ \alpha_{it}^{1}}\Big)-\mathbb{E}\Big(\frac{\hat{\theta}_{it}^{2}}{ \alpha_{it}^{2}}\Big)-  \delta_{e,t}^{\text{max}, FF} \Bigg(\mathbb{E}\Big(\frac{1}{ \alpha_{it}^{1}}\Big)+ \mathbb{E}\Big(\frac{1}{ \alpha_{it}^{2}}\Big) \Bigg)}{2} \nonumber \\
&=& \frac{\mathbb{E}\Big(\frac{\hat{\theta}_{it}^{1}}{ \alpha_{it}^{1}}\Big)+\mathbb{E}\Big(\frac{\hat{\theta}_{it}^{2}}{ \alpha_{it}^{2}}\Big)+\delta_{e,t}^{\text{max}, FF}\Big[\mathbb{E}\Big(\frac{1}{ \alpha_{it}^{2}}\Big)-\mathbb{E}\Big(\frac{1}{ \alpha_{it}^{1}}\Big) \Big]}{2} \nonumber \\
\end{eqnarray}
Where we have substituted $\delta_{m,t}^{\text{max}, FF}$ with its relationship with $\delta_{e,t}^{\text{max}, FF}$ across the $FF$ and the fact that in $FF$ the upper bound is the same as the lower bound. It is worth noting that from Equation \ref{Eq0Thm1} for $\delta_{e,t}^{\text{max}},\delta_{m,t}^{\text{max}}$ in $FF$ we know that Equation \ref{Eq1Thm1} is non-negative (in particular from the first inequality). For an element of $FF(\boldsymbol{\delta}_t)$, say $\delta_{m,t}^{\text{max}, FF}$ and $\delta_{e,t}^{\text{max}, FF}$, we have: 
\begin{align}\label{Eq1Thm1}
\max\{0, \max_{j=1,2} LB_{\delta^{FF},m,t}\} 
&= \min_{j=1,2} UB_{\delta^{FF},m,t} \nonumber \\
&= \mathbb{E}\Big(\frac{\hat{\theta}_{it}^{2}}{ \alpha_{it}^{2}}\Big) + \delta_{e,t}^{\text{max}, FF} \mathbb{E}\Big(\frac{1}{ \alpha_{it}^{2}}\Big) + \delta_{m,t}^{\text{max}, FF} \nonumber \\
&= \mathbb{E}\Big(\frac{\hat{\theta}_{it}^{2}}{ \alpha_{it}^{2}}\Big) + \delta_{e,t}^{\text{max}, FF} \mathbb{E}\Big(\frac{1}{ \alpha_{it}^{2}}\Big) \nonumber \\
&\quad + \frac{\mathbb{E}\Big(\frac{\hat{\theta}_{it}^{1}}{ \alpha_{it}^{1}}\Big) - \mathbb{E}\Big(\frac{\hat{\theta}_{it}^{2}}{ \alpha_{it}^{2}}\Big) - \delta_{e,t}^{\text{max}, FF} \Bigg( \mathbb{E}\Big(\frac{1}{ \alpha_{it}^{1}}\Big) + \mathbb{E}\Big(\frac{1}{ \alpha_{it}^{2}}\Big) \Bigg)}{2} \nonumber \\
&= \frac{\mathbb{E}\Big(\frac{\hat{\theta}_{it}^{1}}{ \alpha_{it}^{1}}\Big) + \mathbb{E}\Big(\frac{\hat{\theta}_{it}^{2}}{ \alpha_{it}^{2}}\Big) + \delta_{e,t}^{\text{max}, FF} \Big[ \mathbb{E}\Big(\frac{1}{ \alpha_{it}^{2}}\Big) - \mathbb{E}\Big(\frac{1}{ \alpha_{it}^{1}}\Big) \Big]}{2}
\end{align}
Where we have substituted $\delta_{m,t}^{\text{max}, FF}$ with its relationship with $\delta_{e,t}^{\text{max}, FF}$ across the $FF$, and the fact that in $FF$ the upper bound is the same as the lower bound. It is worth noting that from Equation \ref{Eq0Thm1} for $\delta_{e,t}^{\text{max}}, \delta_{m,t}^{\text{max}}$ in $FF$, we know that Equation \ref{Eq1Thm1} is non-negative (in particular from the first inequality).

\par 
Similarly in the case that: 
\begin{eqnarray}\label{Eq0bThm1}
    \mathbb{E}\Big(\frac{\hat{\theta}_{it}^{2}}{ \alpha_{it}^{2}}\Big)- \delta_{e,t}^{\text{max}}\mathbb{E}\Big(\frac{1}{ \alpha_{it}^{2}}\Big) &\geq& \delta_{m,t}^{\text{max}}\nonumber \\
    \mathbb{E}\Big(\frac{\hat{\theta}_{it}^{1}}{ \alpha_{it}^{1}}\Big)-\mathbb{E}\Big(\frac{\hat{\theta}_{it}^{2}}{ \alpha_{it}^{2}}\Big)+ \delta_{e,t}^{\text{max}}\Bigg( \mathbb{E}\Big(\frac{1}{ \alpha_{it}^{2}}\Big)-\mathbb{E}\Big(\frac{1}{ \alpha_{it}^{1}}\Big)\Bigg) &\leq& 0 \nonumber \\
    \mathbb{E}\Big(\frac{\hat{\theta}_{it}^{1}}{ \alpha_{it}^{1}}\Big)-\mathbb{E}\Big(\frac{\hat{\theta}_{it}^{2}}{ \alpha_{it}^{2}}\Big)-\delta_{e,t}^{\text{max}}\Bigg(\mathbb{E}\Big(\frac{1}{ \alpha_{it}^{2}}\Big)-\mathbb{E}\Big(\frac{1}{ \alpha_{it}^{1}}\Big) \Bigg) &\leq& 0 \nonumber \\
\end{eqnarray}
We have: 
\begin{eqnarray}\label{Eq2Thm1}
\max\{0,\max_{j=1,2}LB_{\delta^{FF},m,t}\}&=&\min_{j=1,2} = UB_{\delta^{FF},m,t}    \nonumber \\
&=& \frac{\mathbb{E}\Big(\frac{\hat{\theta}_{it}^{1}}{ \alpha_{it}^{1}}\Big)+\mathbb{E}\Big(\frac{\hat{\theta}_{it}^{2}}{ \alpha_{it}^{2}}\Big)-\delta_{e,t}^{\text{max}, FF}\Big[\mathbb{E}\Big(\frac{1}{ \alpha_{it}^{2}}\Big)-\mathbb{E}\Big(\frac{1}{ \alpha_{it}^{1}}\Big) \Big]}{2} \nonumber \\
\end{eqnarray}
Similarly as above  Equation \ref{Eq2Thm1} is non-negative. 
\par 
These two are the bounds in the $FF$ for the cases among the possible combinations of bounds that have a non empty $FF$.

Since we imposed in the two bounds  being in the $FF$  as a function of $\delta_{e,t}^{\text{max}, FF}$ we now need to calculate the union of the points defined by Equations \ref{Eq1Thm1}  and \ref{Eq2Thm1} across the different possible values of $\delta_{e,t}^{\text{max}, FF} \in [0,+\infty)$. 

The behavior of the union will depend on what happens with the sign of $\mathbb{E}\Big(\frac{1}{ \alpha_{it}^{2}}\Big)-\mathbb{E}\Big(\frac{1}{ \alpha_{it}^{1}}\Big)$ in connection to the conditions from Equations \ref{Eq0Thm1} and \ref{Eq0bThm1} and with the sign of $\mathbb{E}\Big(\frac{\hat{\theta}_{it}^{1}}{ \alpha_{it}^{1}}\Big)-\mathbb{E}\Big(\frac{\hat{\theta}_{it}^{2}}{ \alpha_{it}^{2}}\Big)$. These will determine both if the union across the $\boldsymbol{\delta}_t$ is an increasing or decreasing sequence (due to $\mathbb{E}\Big(\frac{1}{ \alpha_{it}^{2}}\Big)-\mathbb{E}\Big(\frac{1}{ \alpha_{it}^{1}}\Big)$ which multiplies $\delta_{m,t}^{\text{max}}$) and which world we are in (which of the two bounds of the falsifiable bounds are the ones that we should consider). 
\par 
Now note that when  $\mathbb{E}\Big(\frac{1}{ \alpha_{it}^{2}}\Big)-\mathbb{E}\Big(\frac{1}{ \alpha_{it}^{1}}\Big)\leq 0$ we have  $ \delta_{e,t}^{\text{max}, FF}\Big[\mathbb{E}\Big(\frac{1}{ \alpha_{it}^{2}}\Big)-\mathbb{E}\Big(\frac{1}{ \alpha_{it}^{1}}\Big) \Big]\leq 0$  and vice-versa when  $\mathbb{E}\Big(\frac{1}{ \alpha_{it}^{2}}\Big)-\mathbb{E}\Big(\frac{1}{ \alpha_{it}^{1}}\Big)> 0$. 
\par 
\begin{enumerate}
    \item  \textbf{We first discuss the case when $\mathbb{E}\Big(\frac{1}{ \alpha_{it}^{2}}\Big)-\mathbb{E}\Big(\frac{1}{ \alpha_{it}^{1}}\Big)\leq 0$ and    $\mathbb{E}\Big(\frac{\hat{\theta}_{it}^{1}}{ \alpha_{it}^{1}}\Big)-\mathbb{E}\Big(\frac{\hat{\theta}_{it}^{2}}{ \alpha_{it}^{2}}\Big)\geq 0$} 
In this case, when $\delta_{e,t}^{\text{max}, FF}=0$ we have from equation \ref{Eq0Thm1}:  
\begin{eqnarray*}
\frac{\mathbb{E}\Big(\frac{\hat{\theta}_{it}^{1}}{ \alpha_{it}^{1}}\Big)+\mathbb{E}\Big(\frac{\hat{\theta}_{it}^{2}}{ \alpha_{it}^{2}}\Big)}{2} &\geq& 0 \\
    \mathbb{E}\Big(\frac{\hat{\theta}_{it}^{1}}{ \alpha_{it}^{1}}\Big)-\mathbb{E}\Big(\frac{\hat{\theta}_{it}^{2}}{ \alpha_{it}^{2}}\Big) &\geq& 0 \nonumber \\
    \mathbb{E}\Big(\frac{\hat{\theta}_{it}^{1}}{ \alpha_{it}^{1}}\Big)-\mathbb{E}\Big(\frac{\hat{\theta}_{it}^{2}}{ \alpha_{it}^{2}}\Big) &\geq& 0 \nonumber \\
\end{eqnarray*}
Since $\mathbb{E}\Big(\frac{\hat{\theta}_{it}^{1}}{ \alpha_{it}^{1}}\Big)-\mathbb{E}\Big(\frac{\hat{\theta}_{it}^{2}}{ \alpha_{it}^{2}}\Big)\geq 0$ all the conditions are satisfied and thus $\frac{\mathbb{E}\Big(\frac{\hat{\theta}_{it}^{1}}{ \alpha_{it}^{1}}\Big)+\mathbb{E}\Big(\frac{\hat{\theta}_{it}^{2}}{ \alpha_{it}^{2}}\Big)}{2}$ is on the FAS. 
\par 
Now pick a generic $\delta_{e,t}^{\text{max}, FF}$. For such to pick the expression from \ref{Eq1Thm1} we need to full fill: 
\begin{eqnarray}
    \frac{\mathbb{E}\Big(\frac{\hat{\theta}_{it}^{1}}{ \alpha_{it}^{1}}\Big)+\mathbb{E}\Big(\frac{\hat{\theta}_{it}^{2}}{ \alpha_{it}^{2}}\Big)+\delta_{e,t}^{\text{max}, FF}\Big[\mathbb{E}\Big(\frac{1}{ \alpha_{it}^{2}}\Big)-\mathbb{E}\Big(\frac{1}{ \alpha_{it}^{1}}\Big) \Big]}{2} &\geq&0 \nonumber \\
    \mathbb{E}\Big(\frac{\hat{\theta}_{it}^{1}}{ \alpha_{it}^{1}}\Big)-\mathbb{E}\Big(\frac{\hat{\theta}_{it}^{2}}{ \alpha_{it}^{2}}\Big)+ \delta_{e,t}^{\text{max}, FF}\Bigg( \mathbb{E}\Big(\frac{1}{ \alpha_{it}^{2}}\Big)-\mathbb{E}\Big(\frac{1}{ \alpha_{it}^{1}}\Big)\Bigg) &\geq& 0 \nonumber \\
    \mathbb{E}\Big(\frac{\hat{\theta}_{it}^{1}}{ \alpha_{it}^{1}}\Big)-\mathbb{E}\Big(\frac{\hat{\theta}_{it}^{2}}{ \alpha_{it}^{2}}\Big)-\delta_{e,t}^{\text{max}, FF}\Bigg(\mathbb{E}\Big(\frac{1}{ \alpha_{it}^{2}}\Big)-\mathbb{E}\Big(\frac{1}{ \alpha_{it}^{1}}\Big) \Bigg) &\geq& 0 \nonumber \\
\end{eqnarray}
It is worth noting that when $\mathbb{E}\Big(\frac{1}{ \alpha_{it}^{2}}\Big)-\mathbb{E}\Big(\frac{1}{ \alpha_{it}^{1}}\Big)\leq 0$ and    $\mathbb{E}\Big(\frac{\hat{\theta}_{it}^{1}}{ \alpha_{it}^{1}}\Big)-\mathbb{E}\Big(\frac{\hat{\theta}_{it}^{2}}{ \alpha_{it}^{2}}\Big)\geq 0$ there will eventually be a  $\delta_{e,t}^{\text{max}, FF}$ that will be the last one that satisfies  either the first inequality or the second inequality. Namely,
\begin{eqnarray}
\delta_{e,t}^{\text{max}, FF,b}&\equiv& \min\Bigg\{ \frac{\mathbb{E}\Big(\frac{\hat{\theta}_{it}^{1}}{ \alpha_{it}^{1}}\Big)+\mathbb{E}\Big(\frac{\hat{\theta}_{it}^{2}}{ \alpha_{it}^{2}}\Big)}{\mathbb{E}\Big(\frac{1}{ \alpha_{it}^{1}}\Big)-\mathbb{E}\Big(\frac{1}{ \alpha_{it}^{2}}\Big)}; \frac{\mathbb{E}\Big(\frac{\hat{\theta}_{it}^{1}}{ \alpha_{it}^{1}}\Big)-\mathbb{E}\Big(\frac{\hat{\theta}_{it}^{2}}{ \alpha_{it}^{2}}\Big)}{\mathbb{E}\Big(\frac{1}{ \alpha_{it}^{1}}\Big)-\mathbb{E}\Big(\frac{1}{ \alpha_{it}^{2}}\Big)} \Bigg\}  \nonumber \\ 
\end{eqnarray}
Thus, we  take the union across $\delta_{e,t}^{\text{max}, FF}$ up to $\delta_{e,t}^{\text{max}, FF,b}$ of the form of  
$$\frac{\mathbb{E}\Big(\frac{\hat{\theta}_{it}^{1}}{ \alpha_{it}^{1}}\Big)+\mathbb{E}\Big(\frac{\hat{\theta}_{it}^{2}}{ \alpha_{it}^{2}}\Big)+\delta_{e,t}^{\text{max}, FF}\Big[\mathbb{E}\Big(\frac{1}{ \alpha_{it}^{2}}\Big)-\mathbb{E}\Big(\frac{1}{ \alpha_{it}^{1}}\Big) \Big]}{2}$$ 
And then  stop since in this case  there will be no positive $\delta_{e,t}^{\text{max}, FF}$ such that satisfies the third inequality from equation \ref{Eq0bThm1}.   
Then, the falsification adaptative set is: 
\begin{eqnarray*}
 FAS&=& \Bigg(\frac{\mathbb{E}\Big(\frac{\hat{\theta}_{it}^{1}}{ \alpha_{it}^{1}}\Big)+\mathbb{E}\Big(\frac{\hat{\theta}_{it}^{2}}{ \alpha_{it}^{2}}\Big)+\delta_{e,t}^{\text{max}, FF,b}\Big[\mathbb{E}\Big(\frac{1}{ \alpha_{it}^{2}}\Big)-\mathbb{E}\Big(\frac{1}{ \alpha_{it}^{1}}\Big) \Big]}{2}\\
 &,& \frac{\mathbb{E}\Big(\frac{\hat{\theta}_{it}^{1}}{ \alpha_{it}^{1}}\Big)+\mathbb{E}\Big(\frac{\hat{\theta}_{it}^{2}}{ \alpha_{it}^{2}}\Big)}{2}\Bigg)  \\
 &=& \Bigg(\mathbb{E}\Big(\frac{\hat{\theta}_{it}^{2}}{ \alpha_{it}^{2}}\Big),\frac{\mathbb{E}\Big(\frac{\hat{\theta}_{it}^{1}}{ \alpha_{it}^{1}}\Big)+\mathbb{E}\Big(\frac{\hat{\theta}_{it}^{2}}{ \alpha_{it}^{2}}\Big)}{2}\Bigg) 
\end{eqnarray*}
\item  \textbf{When $\mathbb{E}\Big(\frac{1}{ \alpha_{it}^{2}}\Big)-\mathbb{E}\Big(\frac{1}{ \alpha_{it}^{1}}\Big)\leq 0$ and    $\mathbb{E}\Big(\frac{\hat{\theta}_{it}^{1}}{ \alpha_{it}^{1}}\Big)-\mathbb{E}\Big(\frac{\hat{\theta}_{it}^{2}}{ \alpha_{it}^{2}}\Big)< 0$} 
\par 
In this case note that when $\delta_{e,t}^{\text{max}, FF}=0$ we have from equation \ref{Eq0bThm1}:
\begin{eqnarray*}
    \frac{\mathbb{E}\Big(\frac{\hat{\theta}_{it}^{1}}{ \alpha_{it}^{1}}\Big)+\mathbb{E}\Big(\frac{\hat{\theta}_{it}^{2}}{ \alpha_{it}^{2}}\Big)}{2} &\geq& 0 \nonumber \\
\mathbb{E}\Big(\frac{\hat{\theta}_{it}^{1}}{ \alpha_{it}^{1}}\Big)-\mathbb{E}\Big(\frac{\hat{\theta}_{it}^{2}}{ \alpha_{it}^{2}}\Big) &<& 0 \nonumber \\ \mathbb{E}\Big(\frac{\hat{\theta}_{it}^{1}}{ \alpha_{it}^{1}}\Big)-\mathbb{E}\Big(\frac{\hat{\theta}_{it}^{2}}{ \alpha_{it}^{2}}\Big) &<& 0 \nonumber \\
\end{eqnarray*}
Since $\mathbb{E}\Big(\frac{\hat{\theta}_{it}^{1}}{ \alpha_{it}^{1}}\Big)-\mathbb{E}\Big(\frac{\hat{\theta}_{it}^{2}}{ \alpha_{it}^{2}}\Big)< 0$ all the conditions are satisfied and thus the first point in the union is $\frac{\mathbb{E}\Big(\frac{\hat{\theta}_{it}^{1}}{ \alpha_{it}^{1}}\Big)+\mathbb{E}\Big(\frac{\hat{\theta}_{it}^{2}}{ \alpha_{it}^{2}}\Big)}{2}$
\par 
Now pick a generic $\delta_{e,t}^{\text{max}, FF}$. For such to pick the expression from \ref{Eq2Thm1} we need to full fill: 
\begin{eqnarray*}
    \frac{\mathbb{E}\Big(\frac{\hat{\theta}_{it}^{1}}{ \alpha_{it}^{1}}\Big)+\mathbb{E}\Big(\frac{\hat{\theta}_{it}^{2}}{ \alpha_{it}^{2}}\Big)-\delta_{e,t}^{\text{max}, FF}\Big[\mathbb{E}\Big(\frac{1}{ \alpha_{it}^{2}}\Big)-\mathbb{E}\Big(\frac{1}{ \alpha_{it}^{1}}\Big) \Big]}{2} \geq 0 \nonumber \\
    \mathbb{E}\Big(\frac{\hat{\theta}_{it}^{1}}{ \alpha_{it}^{1}}\Big)-\mathbb{E}\Big(\frac{\hat{\theta}_{it}^{2}}{ \alpha_{it}^{2}}\Big)+ \delta_{e,t}^{\text{max}, FF}\Bigg( \mathbb{E}\Big(\frac{1}{ \alpha_{it}^{2}}\Big)-\mathbb{E}\Big(\frac{1}{ \alpha_{it}^{1}}\Big)\Bigg) &\leq& 0 \nonumber \\
    \mathbb{E}\Big(\frac{\hat{\theta}_{it}^{1}}{ \alpha_{it}^{1}}\Big)-\mathbb{E}\Big(\frac{\hat{\theta}_{it}^{2}}{ \alpha_{it}^{2}}\Big)-\delta_{e,t}^{\text{max}, FF}\Bigg(\mathbb{E}\Big(\frac{1}{ \alpha_{it}^{2}}\Big)-\mathbb{E}\Big(\frac{1}{ \alpha_{it}^{1}}\Big) \Bigg) &\leq& 0 \nonumber \\
\end{eqnarray*}
It is worth noting that when $\mathbb{E}\Big(\frac{1}{ \alpha_{it}^{2}}\Big)-\mathbb{E}\Big(\frac{1}{ \alpha_{it}^{1}}\Big)\leq 0$ and    $\mathbb{E}\Big(\frac{\hat{\theta}_{it}^{1}}{ \alpha_{it}^{1}}\Big)-\mathbb{E}\Big(\frac{\hat{\theta}_{it}^{2}}{ \alpha_{it}^{2}}\Big)< 0$ there will eventually be a  $\delta_{e,t}^{\text{max}, FF}$ that will be the last one that satisfies the third equality. Namely:  

\begin{eqnarray}
 \delta_{e,t}^{\text{max}, FF,b}&\equiv& \frac{\mathbb{E}\Big(\frac{\hat{\theta}_{it}^{2}}{ \alpha_{it}^{2}}\Big)-\mathbb{E}\Big(\frac{\hat{\theta}_{it}^{1}}{ \alpha_{it}^{1}}\Big)}{\mathbb{E}\Big(\frac{1}{ \alpha_{it}^{1}}\Big)-\mathbb{E}\Big(\frac{1}{ \alpha_{it}^{2}}\Big)}
\end{eqnarray}
Thus, we  take the union across $\delta_{e,t}^{\text{max}, FF}$ up to $\delta_{e,t}^{\text{max}, FF,b}$ of the form of  
$$\frac{\mathbb{E}\Big(\frac{\hat{\theta}_{it}^{1}}{ \alpha_{it}^{1}}\Big)+\mathbb{E}\Big(\frac{\hat{\theta}_{it}^{2}}{ \alpha_{it}^{2}}\Big)-\delta_{e,t}^{\text{max}, FF}\Big[\mathbb{E}\Big(\frac{1}{ \alpha_{it}^{2}}\Big)-\mathbb{E}\Big(\frac{1}{ \alpha_{it}^{1}}\Big) \Big]}{2}$$ 
And then  stop since in this case there will be no positive $\delta_{e,t}^{\text{max}, FF}$ such that satisfies the second inequality from equation \ref{Eq0Thm1}.   Then, the falsification adaptative set is: 
\begin{eqnarray*}
 FAS= \Bigg( \frac{\mathbb{E}\Big(\frac{\hat{\theta}_{it}^{1}}{ \alpha_{it}^{1}}\Big)+\mathbb{E}\Big(\frac{\hat{\theta}_{it}^{2}}{ \alpha_{it}^{2}}\Big)}{2}, \mathbb{E}\Big(\frac{\hat{\theta}_{it}^{2}}{ \alpha_{it}^{2}}\Big)\Bigg)   
\end{eqnarray*}

\item  \textbf{When $\mathbb{E}\Big(\frac{1}{ \alpha_{it}^{2}}\Big)-\mathbb{E}\Big(\frac{1}{ \alpha_{it}^{1}}\Big)> 0$ and    $\mathbb{E}\Big(\frac{\hat{\theta}_{it}^{1}}{ \alpha_{it}^{1}}\Big)-\mathbb{E}\Big(\frac{\hat{\theta}_{it}^{2}}{ \alpha_{it}^{2}}\Big)\geq 0$} 
In this case, when $\delta_{e,t}^{\text{max}, FF}=0$ we have from equation \ref{Eq0Thm1}:  
\begin{eqnarray*}
\frac{\mathbb{E}\Big(\frac{\hat{\theta}_{it}^{1}}{ \alpha_{it}^{1}}\Big)+\mathbb{E}\Big(\frac{\hat{\theta}_{it}^{2}}{ \alpha_{it}^{2}}\Big)}{2} &\geq& 0 \\
    \mathbb{E}\Big(\frac{\hat{\theta}_{it}^{1}}{ \alpha_{it}^{1}}\Big)-\mathbb{E}\Big(\frac{\hat{\theta}_{it}^{2}}{ \alpha_{it}^{2}}\Big) &\geq& 0 \nonumber \\
    \mathbb{E}\Big(\frac{\hat{\theta}_{it}^{1}}{ \alpha_{it}^{1}}\Big)-\mathbb{E}\Big(\frac{\hat{\theta}_{it}^{2}}{ \alpha_{it}^{2}}\Big) &\geq& 0 \nonumber \\
\end{eqnarray*}
Since $\mathbb{E}\Big(\frac{\hat{\theta}_{it}^{1}}{ \alpha_{it}^{1}}\Big)-\mathbb{E}\Big(\frac{\hat{\theta}_{it}^{2}}{ \alpha_{it}^{2}}\Big)\geq 0$ all the conditions are satisfied and thus $\frac{\mathbb{E}\Big(\frac{\hat{\theta}_{it}^{1}}{ \alpha_{it}^{1}}\Big)+\mathbb{E}\Big(\frac{\hat{\theta}_{it}^{2}}{ \alpha_{it}^{2}}\Big)}{2}$ is on the FAS. 
\par 
Now pick a generic $\delta_{e,t}^{\text{max, FF}}$. For such to pick the expression from \ref{Eq1Thm1} we need to full fill: 
\begin{eqnarray}
    \frac{\mathbb{E}\Big(\frac{\hat{\theta}_{it}^{1}}{ \alpha_{it}^{1}}\Big)+\mathbb{E}\Big(\frac{\hat{\theta}_{it}^{2}}{ \alpha_{it}^{2}}\Big)+\delta_{e,t}^{\text{max, FF}}\Big[\mathbb{E}\Big(\frac{1}{ \alpha_{it}^{2}}\Big)-\mathbb{E}\Big(\frac{1}{ \alpha_{it}^{1}}\Big) \Big]}{2} &\geq&0 \nonumber \\
    \mathbb{E}\Big(\frac{\hat{\theta}_{it}^{1}}{ \alpha_{it}^{1}}\Big)-\mathbb{E}\Big(\frac{\hat{\theta}_{it}^{2}}{ \alpha_{it}^{2}}\Big)+ \delta_{e,t}^{\text{max, FF}}\Bigg( \mathbb{E}\Big(\frac{1}{ \alpha_{it}^{2}}\Big)-\mathbb{E}\Big(\frac{1}{ \alpha_{it}^{1}}\Big)\Bigg) &\geq& 0 \nonumber \\
    \mathbb{E}\Big(\frac{\hat{\theta}_{it}^{1}}{ \alpha_{it}^{1}}\Big)-\mathbb{E}\Big(\frac{\hat{\theta}_{it}^{2}}{ \alpha_{it}^{2}}\Big)-\delta_{e,t}^{\text{max, FF}}\Bigg(\mathbb{E}\Big(\frac{1}{ \alpha_{it}^{2}}\Big)-\mathbb{E}\Big(\frac{1}{ \alpha_{it}^{1}}\Big) \Bigg) &\geq& 0 \nonumber \\
\end{eqnarray}
Note that when $\mathbb{E}\Big(\frac{1}{ \alpha_{it}^{2}}\Big)-\mathbb{E}\Big(\frac{1}{ \alpha_{it}^{1}}\Big)>0$ and    $\mathbb{E}\Big(\frac{\hat{\theta}_{it}^{1}}{ \alpha_{it}^{1}}\Big)-\mathbb{E}\Big(\frac{\hat{\theta}_{it}^{2}}{ \alpha_{it}^{2}}\Big)\geq 0$ there will eventually be a  $\delta_{e,t}^{\text{max, FF}}$ that will be the last one that satisfies  the third inequality. Namely,
\begin{eqnarray}
 \delta_{e,t}^{\text{max, FF, b}}&\equiv& \frac{\mathbb{E}\Big(\frac{\hat{\theta}_{it}^{1}}{ \alpha_{it}^{1}}\Big)-\mathbb{E}\Big(\frac{\hat{\theta}_{it}^{2}}{ \alpha_{it}^{2}}\Big)}{\mathbb{E}\Big(\frac{1}{ \alpha_{it}^{2}}\Big)-\mathbb{E}\Big(\frac{1}{ \alpha_{it}^{1}}\Big)}  \nonumber \\ 
\end{eqnarray}
Thus, we  take the union across $\delta_{e,t}^{\text{max, FF}}$ up to $\delta_{e,t}^{\text{max, FF, b}}$ of the form of  
$$\frac{\mathbb{E}\Big(\frac{\hat{\theta}_{it}^{1}}{ \alpha_{it}^{1}}\Big)+\mathbb{E}\Big(\frac{\hat{\theta}_{it}^{2}}{ \alpha_{it}^{2}}\Big)+\delta_{e,t}^{\text{max, FF}}\Big[\mathbb{E}\Big(\frac{1}{ \alpha_{it}^{2}}\Big)-\mathbb{E}\Big(\frac{1}{ \alpha_{it}^{1}}\Big) \Big]}{2}$$ 
And then  stop since in this case  there will be no positive $\delta_{e,t}^{\text{max, FF}}$ such that satisfies the second inequality from equation \ref{Eq0bThm1}.   
Then, the falsification adaptative set is: 
\begin{eqnarray*}
 FAS 
 &=& \Bigg(\frac{\mathbb{E}\Big(\frac{\hat{\theta}_{it}^{1}}{ \alpha_{it}^{1}}\Big)+\mathbb{E}\Big(\frac{\hat{\theta}_{it}^{2}}{ \alpha_{it}^{2}}\Big)}{2}, \mathbb{E}\Big(\frac{\hat{\theta}_{it}^{1}}{ \alpha_{it}^{1}}\Big)\Bigg) 
\end{eqnarray*}
\item  \textbf{When $\mathbb{E}\Big(\frac{1}{ \alpha_{it}^{2}}\Big)-\mathbb{E}\Big(\frac{1}{ \alpha_{it}^{1}}\Big)> 0$ and    $\mathbb{E}\Big(\frac{\hat{\theta}_{it}^{1}}{ \alpha_{it}^{1}}\Big)-\mathbb{E}\Big(\frac{\hat{\theta}_{it}^{2}}{ \alpha_{it}^{2}}\Big)< 0$} 
By symmetry we can note that: 

\begin{eqnarray*}
 FAS &=& \Bigg( \mathbb{E}\Big(\frac{\hat{\theta}_{it}^{1}}{ \alpha_{it}^{1}}\Big), \frac{\mathbb{E}\Big(\frac{\hat{\theta}_{it}^{1}}{ \alpha_{it}^{1}}\Big)+\mathbb{E}\Big(\frac{\hat{\theta}_{it}^{2}}{ \alpha_{it}^{2}}\Big)}{2}\Bigg) 
\end{eqnarray*}
\end{enumerate}
\end{proof}

\section{Results for More Than Two Variable Inputs}\label{Appmorethan2}

\subsubsection*{Variable-input index}

One alternative is to aggregate variable inputs into a small number of broad categories---for example, a materials bundle and a labor bundle---and then apply the two-input analysis developed in the main text. The relevant inputs should first be partitioned into economically meaningful groups. Next, one computes total expenditure for each category and uses the corresponding expenditure shares as weights to construct a single combined input index for each bundle. A similar aggregation can be carried out for the corresponding input prices. This approach preserves the logic of the two-input framework while reducing dimensionality.

\subsubsection*{The union of pairwise FAS submodels}

Another alternative is to compute the falsification adaptive set (FAS) for every pair of variable inputs and then take the union of the resulting pairwise sets. This procedure generally yields a wider set than the exact FAS that would be obtained from the full multi-input model, but it remains a useful and transparent robustness measure. The resulting identification region collects all average-markup values that are compatible with at least one pairwise falsification frontier.

To illustrate, suppose there are three variable inputs, $V_{it}^{1}$, $V_{it}^{2}$, and $V_{it}^{3}$, and suppose
\[
\mathbb{E}\!\left(\frac{1}{\alpha_{it}^{2}}\right)-\mathbb{E}\!\left(\frac{1}{\alpha_{it}^{1}}\right)\leq 0,
\qquad
\mathbb{E}\!\left(\frac{\hat{\theta}_{it}^{1}}{\alpha_{it}^{1}}\right)-\mathbb{E}\!\left(\frac{\hat{\theta}_{it}^{2}}{\alpha_{it}^{2}}\right)\geq 0,
\]
\[
\mathbb{E}\!\left(\frac{1}{\alpha_{it}^{3}}\right)-\mathbb{E}\!\left(\frac{1}{\alpha_{it}^{2}}\right)\leq 0,
\qquad
\mathbb{E}\!\left(\frac{\hat{\theta}_{it}^{2}}{\alpha_{it}^{2}}\right)-\mathbb{E}\!\left(\frac{\hat{\theta}_{it}^{3}}{\alpha_{it}^{3}}\right)\geq 0,
\]
and
\[
\mathbb{E}\!\left(\frac{1}{\alpha_{it}^{3}}\right)-\mathbb{E}\!\left(\frac{1}{\alpha_{it}^{1}}\right)\leq 0,
\qquad
\mathbb{E}\!\left(\frac{\hat{\theta}_{it}^{1}}{\alpha_{it}^{1}}\right)-\mathbb{E}\!\left(\frac{\hat{\theta}_{it}^{3}}{\alpha_{it}^{3}}\right)\geq 0.
\]

Then the pairwise FAS intervals are, respectively,
\[
FAS_{1,2}
=
\left(
\mathbb{E}\!\left(\frac{\hat{\theta}_{it}^{2}}{\alpha_{it}^{2}}\right),
\,
\frac{
\mathbb{E}\!\left(\frac{\hat{\theta}_{it}^{1}}{\alpha_{it}^{1}}\right)
+
\mathbb{E}\!\left(\frac{\hat{\theta}_{it}^{2}}{\alpha_{it}^{2}}\right)
}{2}
\right),
\]
\[
FAS_{2,3}
=
\left(
\mathbb{E}\!\left(\frac{\hat{\theta}_{it}^{3}}{\alpha_{it}^{3}}\right),
\,
\frac{
\mathbb{E}\!\left(\frac{\hat{\theta}_{it}^{3}}{\alpha_{it}^{3}}\right)
+
\mathbb{E}\!\left(\frac{\hat{\theta}_{it}^{2}}{\alpha_{it}^{2}}\right)
}{2}
\right),
\]
and
\[
FAS_{1,3}
=
\left(
\mathbb{E}\!\left(\frac{\hat{\theta}_{it}^{3}}{\alpha_{it}^{3}}\right),
\,
\frac{
\mathbb{E}\!\left(\frac{\hat{\theta}_{it}^{1}}{\alpha_{it}^{1}}\right)
+
\mathbb{E}\!\left(\frac{\hat{\theta}_{it}^{3}}{\alpha_{it}^{3}}\right)
}{2}
\right).
\]

If, in addition,
\[
\mathbb{E}\!\left(\frac{\hat{\theta}_{it}^{3}}{\alpha_{it}^{3}}\right)
\leq
\mathbb{E}\!\left(\frac{\hat{\theta}_{it}^{2}}{\alpha_{it}^{2}}\right)
\leq
\mathbb{E}\!\left(\frac{\hat{\theta}_{it}^{1}}{\alpha_{it}^{1}}\right),
\]
then these intervals overlap, and their union is the connected interval
\[
\bigcup_{k=1}^{3} FAS_k
=
\left(
\mathbb{E}\!\left(\frac{\hat{\theta}_{it}^{3}}{\alpha_{it}^{3}}\right),
\,
\frac{
\mathbb{E}\!\left(\frac{\hat{\theta}_{it}^{1}}{\alpha_{it}^{1}}\right)
+
\mathbb{E}\!\left(\frac{\hat{\theta}_{it}^{2}}{\alpha_{it}^{2}}\right)
}{2}
\right).
\]

More generally, let the pairwise FAS intervals be denoted by
\[
FAS_k=(FAS_{k,lb},FAS_{k,ub}),
\qquad k=1,\dots,K,
\]
where $K$ is the number of distinct input pairs. If the family $\{FAS_k\}_{k=1}^{K}$ is connected---for example, if the intervals overlap pairwise or form an overlapping chain---then
\[
\bigcup_{k=1}^{K} FAS_k
=
\left(
\min_{1\leq k\leq K} FAS_{k,lb},
\,
\max_{1\leq k\leq K} FAS_{k,ub}
\right).
\]
\section{Estimation and Inference Details}\label{AppEInf}

In this Appendix we collect results related to statistical inference for the FAS bounds.

\subsection{Hadamard directional differentiability}

In this subsection we study the differentiability properties of the maps entering our inferential procedure. Since our inferential argument relies on \cite{fang2019inference}, we begin by recalling the relevant notion of Hadamard directional differentiability.

\subsubsection*{Hadamard directional differentiability: definition and illustrative examples}

Let $D$ and $E$ be normed spaces, let $D_{\phi}\subseteq D$, and let
\[
\phi:D_{\phi}\to E.
\]
Fix $\theta\in D_{\phi}$ and let $D_{0}\subseteq D$ be a set of directions.

\begin{definition}
The map $\phi$ is said to be \emph{Hadamard directionally differentiable} at $\theta$ tangentially to $D_{0}$ if there exists a continuous map
\[
\phi'_{\theta}:D_{0}\to E
\]
such that
\[
\lim_{n\to\infty}
\left\|
\frac{\phi(\theta+t_{n}h_{n})-\phi(\theta)}{t_{n}}
-
\phi'_{\theta}(h)
\right\|_{E}
=0
\]
for every sequence $t_{n}\downarrow 0$ and every sequence $h_{n}\to h\in D_{0}$ satisfying $\theta+t_{n}h_{n}\in D_{\phi}$ for all $n$.
\end{definition}

The distinction between full Hadamard differentiability and Hadamard directional differentiability is that, in the former case, the derivative map must be linear and continuous. Proposition 2.1 in \cite{fang2019inference} shows that if the directional derivative is linear, then full Hadamard differentiability is obtained.

\paragraph{Example 1: a fully differentiable map.}
Let
\[
\phi:\mathbb{R}^{2}\to\mathbb{R},
\qquad
\phi(x,y)=x+y.
\]
Fix $\theta=(x_0,y_0)\in\mathbb{R}^{2}$ and let $h=(h_1,h_2)\in\mathbb{R}^{2}$. For any $t_n\downarrow 0$ and any $h_n=(h_{1n},h_{2n})\to (h_1,h_2)$,
\[
\frac{\phi(\theta+t_n h_n)-\phi(\theta)}{t_n}
=
\frac{(x_0+t_n h_{1n})+(y_0+t_n h_{2n})-(x_0+y_0)}{t_n}
=
h_{1n}+h_{2n}.
\]
Taking limits yields
\[
\phi'_{\theta}(h_1,h_2)=h_1+h_2.
\]
Since this derivative is linear and continuous, $\phi$ is fully Hadamard differentiable at every point.

\paragraph{Example 2: a map that is only directionally differentiable at the kink.}
Let
\[
\psi:\mathbb{R}^{2}\to\mathbb{R},
\qquad
\psi(x,y)=\max\{x,y\}.
\]
Fix $\theta=(x_0,y_0)\in\mathbb{R}^{2}$.

\medskip
\noindent
\textbf{Case 1: $x_0>y_0$.}
In a neighborhood of $\theta$, the maximum is attained by the first coordinate, so locally
\[
\psi(x,y)=x.
\]
Hence
\[
\psi'_{\theta}(h_1,h_2)=h_1.
\]
This derivative is linear, so $\psi$ is fully Hadamard differentiable at such points.

\medskip
\noindent
\textbf{Case 2: $x_0<y_0$.}
Similarly, in a neighborhood of $\theta$, the maximum is attained by the second coordinate, so
\[
\psi(x,y)=y,
\qquad
\psi'_{\theta}(h_1,h_2)=h_2.
\]
Again, $\psi$ is fully Hadamard differentiable at such points.

\medskip
\noindent
\textbf{Case 3: $x_0=y_0$.}
This is the kink point. Let $\theta=(a,a)$. Then
\[
\psi(\theta+t_n h_n)
=
\max\{a+t_n h_{1n},\,a+t_n h_{2n}\},
\]
and therefore
\[
\frac{\psi(\theta+t_n h_n)-\psi(\theta)}{t_n}
=
\max\{h_{1n},h_{2n}\}.
\]
Since $h_n\to h$, we obtain
\[
\psi'_{\theta}(h_1,h_2)=\max\{h_1,h_2\}.
\]
This derivative is continuous but not linear. Thus, at points on the diagonal $x_0=y_0$, $\psi$ is Hadamard directionally differentiable but not fully Hadamard differentiable.

\paragraph{Example 3: a map that is not even directionally differentiable at the boundary.}
Let
\[
\varphi:\mathbb{R}^{2}\to\mathbb{R},
\qquad
\varphi(x,y)=\mathbb{I}(x-y\leq 0).
\]
If $x_0<y_0$ or $x_0>y_0$, then $\varphi$ is locally constant, and hence fully Hadamard differentiable with derivative zero. However, at points on the boundary $x_0=y_0$, the map fails to be Hadamard directionally differentiable. Indeed, if $\theta=(a,a)$ and $h=(1,0)$, then for any sequence $t_n\downarrow 0$,
\[
\varphi(\theta+t_n h)=\varphi(a+t_n,a)=0,
\qquad
\varphi(\theta)=1,
\]
so
\[
\frac{\varphi(\theta+t_n h)-\varphi(\theta)}{t_n}
=
\frac{0-1}{t_n}
=
-\frac{1}{t_n}\to -\infty.
\]
Hence the difference quotient does not converge to a finite limit in $\mathbb{R}$.

\subsubsection*{Local differentiability of the FAS endpoint maps}

In our context, define
\[
\theta_t^{j,j'}
\equiv
\big(s_t^j,s_t^{j'},m_t^j,m_t^{j'}\big)\in\mathbb{R}^4,
\]
and denote by
\[
\Delta s_t^{j,j'} \equiv s_t^{j'}-s_t^j,
\qquad
\Delta m_t^{j,j'} \equiv m_t^j-m_t^{j'}.
\]

Using indicators, the lower FAS bound can be written as
\begin{align}
    FAS_{lbt}^{j,j'}
    &=
    m_{t}^{j'} \cdot \mathbb{I}\!\left( \Delta s_t^{j,j'} \leq 0 \right)\mathbb{I}\!\left( \Delta m_t^{j,j'} \geq 0 \right)
    + \left( \frac{m_{t}^{j} + m_{t}^{j'}}{2} \right) \cdot \mathbb{I}\!\left( \Delta s_t^{j,j'} \leq 0 \right)\Big( 1 - \mathbb{I}\!\left( \Delta m_t^{j,j'} \geq 0 \right)\Big) \nonumber \\
    &\quad
    + \left( \frac{m_{t}^{j} + m_{t}^{j'}}{2} \right)\cdot \Big(1 - \mathbb{I}\!\left( \Delta s_t^{j,j'} \leq 0 \right)\Big)\mathbb{I}\!\left( \Delta m_t^{j,j'} \geq 0 \right)
    + m_{t}^{j}\cdot \Big(1 - \mathbb{I}\!\left( \Delta s_t^{j,j'} \leq 0 \right)\Big)\Big(1 - \mathbb{I}\!\left( \Delta m_t^{j,j'} \geq 0 \right)\Big).
    \label{eq:FAS_lb_local}
\end{align}
The upper bound is analogously given by
\begin{align}
    FAS_{ubt}^{j,j'}
    &=
    \left( \frac{m_{t}^{j} + m_{t}^{j'}}{2} \right)\cdot \mathbb{I}\!\left( \Delta s_t^{j,j'} \leq 0 \right)\mathbb{I}\!\left( \Delta m_t^{j,j'} \geq 0 \right)
    + m_{t}^{j'}\cdot \mathbb{I}\!\left( \Delta s_t^{j,j'} \leq 0 \right)\Big( 1 - \mathbb{I}\!\left( \Delta m_t^{j,j'} \geq 0 \right)\Big) \nonumber \\
    &\quad
    + m_{t}^{j}\cdot \Big(1 - \mathbb{I}\!\left( \Delta s_t^{j,j'} \leq 0 \right)\Big)\mathbb{I}\!\left( \Delta m_t^{j,j'} \geq 0 \right)
    + \left( \frac{m_{t}^{j} + m_{t}^{j'}}{2} \right)\cdot \Big(1 - \mathbb{I}\!\left( \Delta s_t^{j,j'} \leq 0 \right)\Big)\Big(1 - \mathbb{I}\!\left( \Delta m_t^{j,j'} \geq 0 \right)\Big).
    \label{eq:FAS_ub_local}
\end{align}

As the previous example shows, the scalar indicator map fails to be Hadamard directionally differentiable at boundary points. Therefore, a global differentiability argument based directly on the indicator representation is not available. The correct argument is local.

\begin{proposition}
Fix $(j,j')$ and $t$. Suppose that at the true population value
\[
\Delta s_t^{j,j'}\neq 0
\qquad\text{and}\qquad
\Delta m_t^{j,j'}\neq 0.
\]
Then the maps
\[
\theta_t^{j,j'} \mapsto FAS_{lb}^{j,j'}
\qquad\text{and}\qquad
\theta_t^{j,j'} \mapsto FAS_{ub}^{j,j'}
\]
are fully Hadamard differentiable at $\theta_t^{j,j'}$.
\end{proposition}
The conditions
\[
s_t^{j'}-s_t^j\neq 0
\qquad\text{and}\qquad
m_t^j-m_t^{j'}\neq 0
\]
 have a natural economic interpretation. Recall that
\[
s_t^j=\mathbb{E}\!\left(\frac{1}{\alpha_{it}^j}\right)
\]
is the average inverse revenue share of input $j$, while
\[
m_t^j=\mathbb{E}\!\left(\frac{\hat{\theta}_{it}^j}{\alpha_{it}^j}\right)
\]
is the average estimated markup recovered from input $j$. Hence,
\[
s_t^{j'}-s_t^{j}\neq 0
\]
means that the two inputs do not have the same average inverse revenue share, whereas
\[
m_t^{j}-m_t^{j'}\neq 0
\]
means that the two input-specific estimated markups do not coincide on average. 
\begin{proof}
Let
\[
\delta_s \equiv \Delta s_t^{j,j'},
\qquad
\delta_m \equiv \Delta m_t^{j,j'}.
\]
By assumption, $\delta_s\neq 0$ and $\delta_m\neq 0$. Define
\[
\eta \equiv \frac{1}{4}\min\{|\delta_s|,|\delta_m|\}>0.
\]
Then any point
\[
\tilde{\theta}
=
(\tilde s^j,\tilde s^{j'},\tilde m^j,\tilde m^{j'})
\]
satisfying
\[
\|\tilde{\theta}-\theta_t^{j,j'}\|_\infty < \eta
\]
also satisfies
\[
\operatorname{sign}(\tilde s^{j'}-\tilde s^j)=\operatorname{sign}(\delta_s)
\qquad\text{and}\qquad
\operatorname{sign}(\tilde m^{j}-\tilde m^{j'})=\operatorname{sign}(\delta_m).
\]
Hence the active regime in \eqref{eq:FAS_lb_local}--\eqref{eq:FAS_ub_local} is constant on the open neighborhood
\[
\mathcal{N}_\eta(\theta_t^{j,j'})
=
\left\{\tilde{\theta}:\|\tilde{\theta}-\theta_t^{j,j'}\|_\infty<\eta\right\}.
\]

Therefore, on $\mathcal{N}_\eta(\theta_t^{j,j'})$, both $FAS_{lb}^{j,j'}$ and $FAS_{ub}^{j,j'}$ coincide with affine maps. Affine maps are fully Hadamard differentiable. Hence, for any sequence $r_n\downarrow 0$ and any sequence $h_n\to h$ in $\mathbb{R}^4$, one has eventually
\[
\theta_t^{j,j'}+r_n h_n \in \mathcal{N}_\eta(\theta_t^{j,j'}),
\]
so that the difference quotient is exactly the difference quotient of the corresponding affine map. Therefore it converges to the associated linear derivative.
\end{proof}

\paragraph{Directional derivatives.}
Write a generic direction as
\[
h=(h_{s^j},h_{s^{j'}},h_{m^j},h_{m^{j'}})\in\mathbb{R}^4.
\]
Under the local strict-inequality condition, the directional derivatives of the FAS bounds are
\[
(FAS_{lbt}^{j,j'})'_{\theta}(h)=
\begin{cases}
h_{m^{j'}}, & \text{if } \Delta s_t^{j,j'}<0 \text{ and } \Delta m_t^{j,j'}>0,\\[4pt]
\dfrac{h_{m^j}+h_{m^{j'}}}{2}, & \text{if } \Delta s_t^{j,j'}<0 \text{ and } \Delta m_t^{j,j'}<0,\\[8pt]
\dfrac{h_{m^j}+h_{m^{j'}}}{2}, & \text{if } \Delta s_t^{j,j'}>0 \text{ and } \Delta m_t^{j,j'}>0,\\[8pt]
h_{m^{j}}, & \text{if } \Delta s_t^{j,j'}>0 \text{ and } \Delta m_t^{j,j'}<0,
\end{cases}
\]
and
\[
(FAS_{ubt}^{j,j'})'_{\theta}(h)=
\begin{cases}
\dfrac{h_{m^j}+h_{m^{j'}}}{2}, & \text{if } \Delta s_t^{j,j'}<0 \text{ and } \Delta m_t^{j,j'}>0,\\[8pt]
h_{m^{j'}}, & \text{if } \Delta s_t^{j,j'}<0 \text{ and } \Delta m_t^{j,j'}<0,\\[4pt]
h_{m^{j}}, & \text{if } \Delta s_t^{j,j'}>0 \text{ and } \Delta m_t^{j,j'}>0,\\[4pt]
\dfrac{h_{m^j}+h_{m^{j'}}}{2}, & \text{if } \Delta s_t^{j,j'}>0 \text{ and } \Delta m_t^{j,j'}<0.
\end{cases}
\]

Equivalently, once the regime is fixed, the derivative depends only on perturbations of the markup moments and not on perturbations of the inverse-share moments. The latter determine which regime is active, whereas the value of the FAS endpoint within a fixed regime depends only on $(m_t^j,m_t^{j'})$.

\subsubsection*{Verification of the assumptions in \cite{fang2019inference} for the FAS endpoints}

We now verify that the FAS endpoint maps fall within the scope of the inferential framework in \cite{fang2019inference}. In particular, we show that Assumptions 2.1, 2.2, 3.1, 3.2 and 3.3 in \cite{fang2019inference} are satisfied, so that Corollary 3.2 in that paper can be used to obtain consistent critical values.

\paragraph{Step 1: finite-dimensional parameterization.}
Fix $(t,j,j')$ and define
\[
\theta_0 \equiv \theta_t^{j,j'}
=
\big(s_t^j,s_t^{j'},m_t^j,m_t^{j'}\big)\in\mathbb{R}^4.
\]
Let
\[
\phi_{lbt}(\theta_0)=FAS_{lbt}^{j,j'},
\qquad
\phi_{ubt}(\theta_0)=FAS_{ubt}^{j,j'}.
\]
For each endpoint map separately, we take
\[
D=\mathbb{R}^4,
\qquad
E=\mathbb{R},
\qquad
r_n=\sqrt{n}.
\]
Since $\mathbb{R}^4$ and $\mathbb{R}$ are Banach spaces, Assumption 2.1(i) in \cite{fang2019inference} is immediate.

\paragraph{Step 2: local differentiability under regime separation.}
Define the population differences
\[
\Delta s_0 \equiv s_t^{j'}-s_t^j,
\qquad
\Delta m_0 \equiv m_t^j-m_t^{j'}.
\]
Suppose
\[
\Delta s_0\neq 0
\qquad\text{and}\qquad
\Delta m_0\neq 0.
\]
Then the active regime in the piecewise FAS formula is constant on a neighborhood of $\theta_0$.\footnote{This is the key local point. The FAS map is globally piecewise affine and therefore globally nonsmooth, but differentiability is a \emph{local} property. If $\Delta s_0\neq 0$ and $\Delta m_0\neq 0$, then $\theta_0$ lies strictly inside one regime, rather than on a switching boundary. Consequently, there exists an open neighborhood of $\theta_0$ in which the same branch of the piecewise formula remains active, so the map coincides locally with a single affine function. In that sense the map is fully Hadamard differentiable \emph{at} $\theta_0$, even though it fails to be globally differentiable because different affine branches meet at the boundaries $\Delta s_0=0$ or $\Delta m_0=0$. This is entirely analogous to the map $(x,y)\mapsto \max\{x,y\}$: it is globally nonsmooth, but fully differentiable at every point with $x\neq y$, and only directionally differentiable on the diagonal. See Definition 2.1 and Proposition 2.1 in \cite{fang2019inference}.}

Hence each endpoint map reduces locally to a single affine function, and therefore is fully Hadamard differentiable at $\theta_0$. Since full Hadamard differentiability implies Hadamard directional differentiability, Assumption 2.1(ii) in \cite{fang2019inference} also holds. Moreover, because the derivative map is linear and defined on all of $\mathbb{R}^4$, Assumption 2.1(iii) is immediate.

\paragraph{Verification of Assumption 2.2.}
Let $\hat\theta_n$ denote the estimator of $\theta_0$, and suppose that
\[
\sqrt{n}(\hat\theta_n-\theta_0)\rightsquigarrow G_0
\]
for a Gaussian random vector $G_0\in\mathbb{R}^4$. Then Assumption 2.2(i) in \cite{fang2019inference} holds with $r_n=\sqrt n$. Since the parameter space is finite-dimensional, tightness and support inclusion in $D_0=D=\mathbb{R}^4$ are automatic, so Assumption 2.2(ii) also holds.

\paragraph{Verification of Assumptions 3.1 and 3.2.}
Let $\hat\theta_n^*$ be a bootstrap version of $\hat\theta_n$ satisfying
\[
\sup_{f\in BL_1(\mathbb{R}^4)}
\left|
E^*\!\left[f\!\left(\sqrt n(\hat\theta_n^*-\hat\theta_n)\right)\right]
-
E[f(G_0)]
\right|
=o_p(1),
\]
where $E^*$ denotes bootstrap expectation conditional on the sample. Then Assumption 3.1 in \cite{fang2019inference} holds. Since the bootstrap statistic is finite-dimensional and measurable under standard bootstrap constructions, Assumption 3.2 is also satisfied.

\paragraph{Verification of Assumption 3.3.}
Define the sample analog differences
\[
\hat\Delta s \equiv \hat s_t^{j'}-\hat s_t^j,
\qquad
\hat\Delta m \equiv \hat m_t^j-\hat m_t^{j'}.
\]
Construct the estimator $\hat\phi_n'$ of the directional derivative by selecting the derivative corresponding to the sample regime. For the lower endpoint:
\[
\hat\phi'_{lb,n}(h)=
\begin{cases}
h_{m^{j'}}, & \hat\Delta s<0 \text{ and } \hat\Delta m>0,\\[4pt]
\dfrac{h_{m^j}+h_{m^{j'}}}{2}, & \hat\Delta s<0 \text{ and } \hat\Delta m<0,\\[8pt]
\dfrac{h_{m^j}+h_{m^{j'}}}{2}, & \hat\Delta s>0 \text{ and } \hat\Delta m>0,\\[8pt]
h_{m^{j}}, & \hat\Delta s>0 \text{ and } \hat\Delta m<0,
\end{cases}
\]
and analogously for the upper endpoint.

Under consistency of $\hat\theta_n$ and the regime-separation condition,
\[
P\!\left(
\operatorname{sign}(\hat\Delta s)=\operatorname{sign}(\Delta s_0),
\ \operatorname{sign}(\hat\Delta m)=\operatorname{sign}(\Delta m_0)
\right)\to 1.
\]
On this event, the sample-selected derivative coincides exactly with the population derivative for all $h\in\mathbb{R}^4$. Therefore, for every compact set $K\subset\mathbb{R}^4$ and every $\delta>0$,
\[
\sup_{h\in K^\delta}
\big|
\hat\phi_n'(h)-\phi'_{\theta_0}(h)
\big|
=0
\]
with probability approaching one. Hence Assumption 3.3 in \cite{fang2019inference} holds.

\paragraph{Application of Corollary 3.2.}
Since each endpoint map is scalar-valued, Corollary 3.2 in \cite{fang2019inference} applies provided the cdf of $\phi'_{\theta_0}(G_0)$ is continuous and strictly increasing at its $(1-\alpha)$ quantile. Under the preceding conditions, $\phi'_{\theta_0}(G_0)$ is a  transformation of a Gaussian vector. Continuity and strict monotonicity of its cdf follow whenever
\[
Var\!\big(\phi'_{\theta_0}(G_0)\big)>0.
\]
Consequently, if
\[
Var\!\big(\phi'_{\theta_0}(G_0)\big)>0,
\]
the conditional bootstrap critical value
\[
\hat c_{1-\alpha}
=
\inf\left\{
c:\ 
P^*\!\left(
\hat\phi_n'\!\big(\sqrt n(\hat\theta_n^*-\hat\theta_n)\big)\le c
\right)\ge 1-\alpha
\right\}
\]
satisfies
\[
\hat c_{1-\alpha}\xrightarrow{p} c_{1-\alpha},
\]
where $c_{1-\alpha}$ denotes the $(1-\alpha)$ quantile of $\phi'_{\theta_0}(G_0)$.

\paragraph{Conclusion.}
For each fixed $(t,j,j')$, the lower and upper FAS endpoint maps satisfy Assumptions 2.1, 2.2, 3.1, 3.2 and 3.3 in \cite{fang2019inference} under the following sufficient conditions:
\begin{enumerate}
    \item \emph{Regime separation:}
    \[
    s_t^{j'}-s_t^j\neq 0
    \qquad\text{and}\qquad
    m_t^j-m_t^{j'}\neq 0.
    \]
    \item \emph{Asymptotic linearity / CLT:}
    \[
    \sqrt n(\hat\theta_n-\theta_0)\rightsquigarrow G_0,
    \]
    with $G_0$ Gaussian in $\mathbb{R}^4$.
    \item \emph{Bootstrap validity for the first-stage estimator:}
    the conditional law of $\sqrt n(\hat\theta_n^*-\hat\theta_n)$ consistently estimates the law of $G_0$.
    \item \emph{Nondegenerate limit law:}
    \[
    Var\!\big(\phi'_{\theta_0}(G_0)\big)>0.
    \]
\end{enumerate}
Under these conditions, Corollary 3.2 in \cite{fang2019inference} yields consistent critical values for inference on each FAS endpoint.

\subsubsection*{Bootstrap critical values and confidence intervals for the FAS endpoints}

We next specialize Corollary 3.2 in \cite{fang2019inference} to the lower and upper FAS endpoints. Fix $(t,j,j')$ and let
\[
\theta_0
=
\big(s_t^j,s_t^{j'},m_t^j,m_t^{j'}\big)\in\mathbb{R}^4,
\qquad
\hat\theta_n
=
\big(\hat s_t^j,\hat s_t^{j'},\hat m_t^j,\hat m_t^{j'}\big)\in\mathbb{R}^4.
\]
Let $\hat\theta_n^*$ denote a bootstrap draw from the baseline estimators, and define
\[
\mathbb{G}_n^*
\equiv
\sqrt n\big(\hat\theta_n^*-\hat\theta_n\big)\in\mathbb{R}^4.
\]

We also define the sample analog regime differences
\[
\hat\Delta s
\equiv
\hat s_t^{j'}-\hat s_t^j,
\qquad
\hat\Delta m
\equiv
\hat m_t^j-\hat m_t^{j'}.
\]

For any direction
\[
h=(h_{s^j},h_{s^{j'}},h_{m^j},h_{m^{j'}})\in\mathbb{R}^4,
\]
define the estimated directional derivative of the lower endpoint by
\[
\hat\phi'_{lb,n}(h)=
\begin{cases}
h_{m^{j'}}, & \text{if } \hat\Delta s<0 \text{ and } \hat\Delta m>0,\\[4pt]
\dfrac{h_{m^j}+h_{m^{j'}}}{2}, & \text{if } \hat\Delta s<0 \text{ and } \hat\Delta m<0,\\[8pt]
\dfrac{h_{m^j}+h_{m^{j'}}}{2}, & \text{if } \hat\Delta s>0 \text{ and } \hat\Delta m>0,\\[8pt]
h_{m^{j}}, & \text{if } \hat\Delta s>0 \text{ and } \hat\Delta m<0,
\end{cases}
\]
and analogously define the estimated directional derivative of the upper endpoint by
\[
\hat\phi'_{ub,n}(h)=
\begin{cases}
\dfrac{h_{m^j}+h_{m^{j'}}}{2}, & \text{if } \hat\Delta s<0 \text{ and } \hat\Delta m>0,\\[8pt]
h_{m^{j'}}, & \text{if } \hat\Delta s<0 \text{ and } \hat\Delta m<0,\\[4pt]
h_{m^{j}}, & \text{if } \hat\Delta s>0 \text{ and } \hat\Delta m>0,\\[4pt]
\dfrac{h_{m^j}+h_{m^{j'}}}{2}, & \text{if } \hat\Delta s>0 \text{ and } \hat\Delta m<0.
\end{cases}
\]

Let
\[
Z_{lb,n}^*
\equiv
\hat\phi'_{lb,n}\!\left(\mathbb{G}_n^*\right),
\qquad
Z_{ub,n}^*
\equiv
\hat\phi'_{ub,n}\!\left(\mathbb{G}_n^*\right).
\]
Conditional on the sample, these statistics provide the bootstrap approximation to the limiting distributions of
\[
\sqrt n\left(\widehat{FAS}_{lb}^{j,j'}-FAS_{lb}^{j,j'}\right),
\qquad
\sqrt n\left(\widehat{FAS}_{ub}^{j,j'}-FAS_{ub}^{j,j'}\right),
\]
respectively.

For any $\tau\in(0,1)$, define the conditional bootstrap quantiles
\[
\hat c_{lb,\tau}
\equiv
\inf\left\{
c\in\mathbb{R}:
P^*\!\left(Z_{lb,n}^*\le c\right)\ge \tau
\right\},
\]
and
\[
\hat c_{ub,\tau}
\equiv
\inf\left\{
c\in\mathbb{R}:
P^*\!\left(Z_{ub,n}^*\le c\right)\ge \tau
\right\},
\]
where $P^*(\cdot)$ denotes probability under the bootstrap law conditional on the observed sample.

In particular, the $(1-\alpha)$ critical values are
\[
\hat c_{lb,1-\alpha}
=
\inf\left\{
c\in\mathbb{R}:
P^*\!\left(Z_{lb,n}^*\le c\right)\ge 1-\alpha
\right\},
\]
and
\[
\hat c_{ub,1-\alpha}
=
\inf\left\{
c\in\mathbb{R}:
P^*\!\left(Z_{ub,n}^*\le c\right)\ge 1-\alpha
\right\}.
\]

Let
\[
\widehat{FAS}_{lb}^{j,j'}\equiv \phi_{lb}(\hat\theta_n),
\qquad
\widehat{FAS}_{ub}^{j,j'}\equiv \phi_{ub}(\hat\theta_n).
\]
Then an asymptotic lower one-sided $(1-\alpha)$ confidence interval for the lower endpoint is
\[
CI_{lb,\mathrm{lower}}^{1-\alpha}
=
\left[
\widehat{FAS}_{lb}^{j,j'}
-
\frac{\hat c_{lb,1-\alpha}}{\sqrt n},
\,
+\infty
\right),
\]
while an asymptotic upper one-sided $(1-\alpha)$ confidence interval is
\[
CI_{lb,\mathrm{upper}}^{1-\alpha}
=
\left(
-\infty,
\,
\widehat{FAS}_{lb}^{j,j'}
-
\frac{\hat c_{lb,\alpha}}{\sqrt n}
\right].
\]

Analogously, for the upper endpoint:
\[
CI_{ub,\mathrm{lower}}^{1-\alpha}
=
\left[
\widehat{FAS}_{ub}^{j,j'}
-
\frac{\hat c_{ub,1-\alpha}}{\sqrt n},
\,
+\infty
\right),
\]
and
\[
CI_{ub,\mathrm{upper}}^{1-\alpha}
=
\left(
-\infty,
\,
\widehat{FAS}_{ub}^{j,j'}
-
\frac{\hat c_{ub,\alpha}}{\sqrt n}
\right].
\]

Two-sided asymptotic $(1-\alpha)$ confidence interval for the lower endpoint can be defined as
\[
CI_{lb}^{1-\alpha}
=
\left[
\widehat{FAS}_{lb}^{j,j'}
-
\frac{\hat c_{lb,1-\alpha/2}}{\sqrt n},
\,
\widehat{FAS}_{lb}^{j,j'}
-
\frac{\hat c_{lb,\alpha/2}}{\sqrt n}
\right].
\]
Similarly, a two-sided asymptotic $(1-\alpha)$ confidence interval for the upper endpoint is
\[
CI_{ub}^{1-\alpha}
=
\left[
\widehat{FAS}_{ub}^{j,j'}
-
\frac{\hat c_{ub,1-\alpha/2}}{\sqrt n},
\,
\widehat{FAS}_{ub}^{j,j'}
-
\frac{\hat c_{ub,\alpha/2}}{\sqrt n}
\right].
\]

Under: 
(i) regime separation,
\[
s_t^{j'}-s_t^j\neq 0
\qquad\text{and}\qquad
m_t^j-m_t^{j'}\neq 0,
\]
(ii) a Gaussian limit law for $\sqrt n(\hat\theta_n-\theta_0)$, 
(iii) bootstrap consistency for $\sqrt n(\hat\theta_n^*-\hat\theta_n)$, and (iv) nondegeneracy of the scalar limit law---Corollary 3.2 in \cite{fang2019inference} implies
\[
\hat c_{lb,\tau}\xrightarrow{p} c_{lb,\tau},
\qquad
\hat c_{ub,\tau}\xrightarrow{p} c_{ub,\tau},
\]
for every continuity point $\tau$ of the corresponding limiting distributions. Consequently, the confidence intervals above are asymptotically valid for the scalar FAS endpoints. This follows directly from Theorem 3.2 and Corollary 3.2 in \cite{fang2019inference}.

\subsubsection*{Bootstrap Implementation}

We summarize the construction of the bootstrap critical values for the FAS endpoints in the following algorithm.

\begin{algorithm}
\begin{enumerate}
    \item From the observed data
    \[
    \{Q_{it},V_{it}^1,V_{it}^2,K_{it},P_{it}^1,P_{it}^2,P_{it}Q_{it}\}_{i,t}^{N,T},
    \]
    compute the expenditure shares
    \[
    \alpha_{it}^j=\frac{P_{it}^jV_{it}^j}{P_{it}Q_{it}},
    \qquad j=1,2,
    \]
    and, using a chosen production-function estimator, obtain $\hat{\hat{\theta}}_{it}^j$ for each $j=1,2$.

    \item For each time period $t$, compute the sample analogues
    \[
    \hat s_{t,n}^j
    \equiv
    \frac{1}{n_t}\sum_{i=1}^{n_t}\frac{1}{\alpha_{it}^j},
    \qquad
    \hat m_{t,n}^j
    \equiv
    \frac{1}{n_t}\sum_{i=1}^{n_t}\frac{\hat{\hat{\theta}}_{it}^j}{\alpha_{it}^j},
    \qquad j=1,2.
    \]

    \item For the input pair $(j,j')$, compute
    \[
    \hat\Delta s_t^{j,j'}\equiv \hat s_{t,n}^{j'}-\hat s_{t,n}^{j},
    \qquad
    \hat\Delta m_t^{j,j'}\equiv \hat m_{t,n}^{j}-\hat m_{t,n}^{j'}.
    \]

    \item Using the signs of $\hat\Delta s_t^{j,j'}$ and $\hat\Delta m_t^{j,j'}$, determine the active regime and compute the sample FAS endpoints
    \[
    \widehat{FAS}_{lb,t}^{j,j'},
    \qquad
    \widehat{FAS}_{ub,t}^{j,j'}.
    \]

    \item Draw $B$ bootstrap samples by resampling the data with replacement, clustered at the unit level $i$.

    \item For each bootstrap replication $b=1,\dots,B$, re-estimate the production function and obtain the bootstrap elasticities
    \[
    \hat{\hat{\theta}}_{it}^{j,b},
    \qquad
    \hat{\hat{\theta}}_{it}^{j',b}.
    \]

    \item For each $b$ and $t$, compute the bootstrap analogues
    \[
    \hat s_{t,n}^{j,b},
    \qquad
    \hat m_{t,n}^{j,b},
    \qquad
    \hat s_{t,n}^{j',b},
    \qquad
    \hat m_{t,n}^{j',b}.
    \]

    \item For each $b$ and $t$, construct the bootstrap perturbation
    \[
    \mathbb{G}_{t,n}^{*,b}
    =
    \left(
    \sqrt n(\hat m_{t,n}^{j,b}-\hat m_{t,n}^{j}),
    \sqrt n(\hat m_{t,n}^{j',b}-\hat m_{t,n}^{j'})
    \right).
    \]
    Since the active regime is fixed by the signs of $\hat\Delta s_t^{j,j'}$ and $\hat\Delta m_t^{j,j'}$, the derivative depends only on the perturbations of the markup moments. Thus define
    \[
    \hat\phi'_{lb,t,n}\!\left(\mathbb{G}_{t,n}^{*,b}\right)
    =
    \begin{cases}
    \sqrt n(\hat m_{t,n}^{j',b}-\hat m_{t,n}^{j'}), 
    & \text{if } \hat\Delta s_t^{j,j'}<0,\ \hat\Delta m_t^{j,j'}>0,\\[4pt]
    \dfrac{
    \sqrt n(\hat m_{t,n}^{j,b}-\hat m_{t,n}^{j})
    +
    \sqrt n(\hat m_{t,n}^{j',b}-\hat m_{t,n}^{j'})
    }{2},
    & \text{if } \hat\Delta s_t^{j,j'}<0,\ \hat\Delta m_t^{j,j'}<0,\\[8pt]
    \dfrac{
    \sqrt n(\hat m_{t,n}^{j,b}-\hat m_{t,n}^{j})
    +
    \sqrt n(\hat m_{t,n}^{j',b}-\hat m_{t,n}^{j'})
    }{2},
    & \text{if } \hat\Delta s_t^{j,j'}>0,\ \hat\Delta m_t^{j,j'}>0,\\[8pt]
    \sqrt n(\hat m_{t,n}^{j,b}-\hat m_{t,n}^{j}),
    & \text{if } \hat\Delta s_t^{j,j'}>0,\ \hat\Delta m_t^{j,j'}<0,
    \end{cases}
    \]
    and
    \[
    \hat\phi'_{ub,t,n}\!\left(\mathbb{G}_{t,n}^{*,b}\right)
    =
    \begin{cases}
    \dfrac{
    \sqrt n(\hat m_{t,n}^{j,b}-\hat m_{t,n}^{j})
    +
    \sqrt n(\hat m_{t,n}^{j',b}-\hat m_{t,n}^{j'})
    }{2},
    & \text{if } \hat\Delta s_t^{j,j'}<0,\ \hat\Delta m_t^{j,j'}>0,\\[8pt]
    \sqrt n(\hat m_{t,n}^{j',b}-\hat m_{t,n}^{j'}),
    & \text{if } \hat\Delta s_t^{j,j'}<0,\ \hat\Delta m_t^{j,j'}<0,\\[4pt]
    \sqrt n(\hat m_{t,n}^{j,b}-\hat m_{t,n}^{j}),
    & \text{if } \hat\Delta s_t^{j,j'}>0,\ \hat\Delta m_t^{j,j'}>0,\\[4pt]
    \dfrac{
    \sqrt n(\hat m_{t,n}^{j,b}-\hat m_{t,n}^{j})
    +
    \sqrt n(\hat m_{t,n}^{j',b}-\hat m_{t,n}^{j'})
    }{2},
    & \text{if } \hat\Delta s_t^{j,j'}>0,\ \hat\Delta m_t^{j,j'}<0.
    \end{cases}
    \]

    \item Use the empirical distributions
    \[
    \left\{
    \hat\phi'_{lb,t,n}\!\left(\mathbb{G}_{t,n}^{*,b}\right)
    \right\}_{b=1}^{B},
    \qquad
    \left\{
    \hat\phi'_{ub,t,n}\!\left(\mathbb{G}_{t,n}^{*,b}\right)
    \right\}_{b=1}^{B},
    \]
    to estimate the bootstrap quantiles $\hat c_{lb,t,\tau}$ and $\hat c_{ub,t,\tau}$.

    \item Construct the one-sided confidence intervals
    \[
    CI_{lb,t,\mathrm{lower}}^{1-\alpha}
    =
    \left[
    \widehat{FAS}_{lb,t}^{j,j'}
    -
    \frac{\hat c_{lb,t,1-\alpha}}{\sqrt n},
    \,
    +\infty
    \right),
    \]
    and
    \[
    CI_{ub,t,\mathrm{upper}}^{1-\alpha}
    =
    \left(
    -\infty,
    \,
    \widehat{FAS}_{ub,t}^{j,j'}
    -
    \frac{\hat c_{ub,t,\alpha}}{\sqrt n}
    \right].
    \]
\end{enumerate}
\end{algorithm}

\subsubsection{Additional Implementation of the Bootstrap Procedure for the FAS Endpoints in the Empirical Application}

This subsection describes additional implementation details of the bootstrap procedure used to construct confidence intervals for the lower and upper endpoints of the falsification adaptive set (FAS) in the empirical application.

We focus on the pair of variable inputs given by labor and materials. In the notation of the paper, labor corresponds to input $j=1$ and materials correspond to input $j'=2$. For each year $t$, we construct the sample analogues
\[
\hat s_{t,n}^1
=
\frac{1}{n_t}\sum_{i=1}^{n_t}\frac{1}{\alpha_{it}^{1}},
\qquad
\hat s_{t,n}^2
=
\frac{1}{n_t}\sum_{i=1}^{n_t}\frac{1}{\alpha_{it}^{2}},
\]
and
\[
\hat m_{t,n}^1
=
\frac{1}{n_t}\sum_{i=1}^{n_t}\frac{\hat{\hat\theta}_{it}^{1}}{\alpha_{it}^{1}},
\qquad
\hat m_{t,n}^2
=
\frac{1}{n_t}\sum_{i=1}^{n_t}\frac{\hat{\hat\theta}_{it}^{2}}{\alpha_{it}^{2}}.
\]
In the empirical implementation, the input-specific estimated markups implied by the selected production-function estimator are directly available. Accordingly, the objects $\hat m_{t,n}^1$ and $\hat m_{t,n}^2$ are computed as the yearly averages of the labor-based and materials-based markup estimates, respectively.

The yearly regime is determined by the signs of
\[
\hat\Delta s_t^{1,2}
=
\hat s_{t,n}^{2}-\hat s_{t,n}^{1},
\qquad
\hat\Delta m_t^{1,2}
=
\hat m_{t,n}^{1}-\hat m_{t,n}^{2}.
\]
Given these signs, the relevant piecewise formula is selected and the sample FAS endpoints,
\[
\widehat{FAS}_{lb,t}^{1,2}
\qquad\text{and}\qquad
\widehat{FAS}_{ub,t}^{1,2},
\]
are computed year by year.

The bootstrap is implemented by resampling establishments with replacement, clustered at the plant level. Clustering at the plant level preserves the panel structure of the data and is consistent with the fact that the production-function estimator relies on lagged plant-level variables. In each bootstrap replication, the production function is re-estimated from scratch, and the corresponding bootstrap analogues
\[
\hat m_{t,n}^{1,b},
\qquad
\hat m_{t,n}^{2,b}
\]
are recomputed for each year $t$ and bootstrap draw $b$.

The bootstrap perturbation is constructed as
\[
\mathbb{G}_{t,n}^{*,b}
=
\left(
\sqrt{n_t}\big(\hat m_{t,n}^{1,b}-\hat m_{t,n}^{1}\big),
\sqrt{n_t}\big(\hat m_{t,n}^{2,b}-\hat m_{t,n}^{2}\big)
\right).
\]
This choice follows the standard bootstrap convention in which the empirical process is centered at the original estimator. Once the regime is fixed by the signs of $(\hat\Delta s_t^{1,2},\hat\Delta m_t^{1,2})$, the directional derivative of each endpoint depends only on perturbations of the markup moments $(m_t^1,m_t^2)$ and not on perturbations of the inverse-share moments $(s_t^1,s_t^2)$. Therefore, in practice, the derivative is evaluated only on the bootstrap perturbations of the $m$-components.

For each bootstrap replication, the estimated directional derivatives of the lower and upper endpoints are computed using the regime-specific formulas derived in the Appendix. The empirical distributions of these bootstrap derivative draws are then used to estimate the relevant quantiles. Finally, one-sided confidence intervals are constructed as
\[
CI_{lb,t,\mathrm{lower}}^{1-\alpha}
=
\left[
\widehat{FAS}_{lb,t}^{1,2}
-
\frac{\hat c_{lb,t,1-\alpha}}{\sqrt{n_t}},
\,
+\infty
\right),
\]
and
\[
CI_{ub,t,\mathrm{upper}}^{1-\alpha}
=
\left(
-\infty,
\,
\widehat{FAS}_{ub,t}^{1,2}
-
\frac{\hat c_{ub,t,\alpha}}{\sqrt{n_t}}
\right].
\]

\subsection{Pseudo-Stata implementation of the FAS and its confidence intervals}\label{app:pseudostata_fas}

This appendix provides a step-by-step pseudo-Stata implementation of the falsification adaptive set (FAS) for the average markup and its bootstrap confidence intervals in the two-input case. The procedure requires, for each firm-year observation, (i) total revenue and variable-input expenditures to construct revenue shares, and (ii) estimated output elasticities for at least two variable inputs obtained from a production-function estimator chosen by the researcher. The code below is specialized to the two-input case discussed in the main text. It first computes the sample analogs of the moments entering the FAS, then applies the piecewise FAS formula, and finally constructs bootstrap critical values and confidence intervals for the lower and upper FAS endpoints using the local directional derivative approach.

\paragraph{Required inputs.}
For each firm-year observation, the researcher must have:
\begin{itemize}
    \item a firm identifier $id$ and a time identifier $t$;
    \item total sales or revenues $sales$;
    \item expenditures on two variable inputs, denoted $exp1$ and $exp2$;
    \item estimated output elasticities $\hat\theta_{it}^{1}$ and $\hat\theta_{it}^{2}$ obtained from a chosen production-function estimator.
\end{itemize}
From these objects, define the revenue shares
\[
\alpha_{it}^{1} = \frac{exp1_{it}}{sales_{it}},
\qquad
\alpha_{it}^{2} = \frac{exp2_{it}}{sales_{it}},
\]
and the input-specific estimated markups
\[
\hat\mu_{it}^{1} = \frac{\hat\theta_{it}^{1}}{\alpha_{it}^{1}},
\qquad
\hat\mu_{it}^{2} = \frac{\hat\theta_{it}^{2}}{\alpha_{it}^{2}}.
\]
The year-specific moments entering the FAS are
\[
s_t^j \equiv E\!\left[\frac{1}{\alpha_{it}^{j}}\right],
\qquad
m_t^j \equiv E\!\left[\frac{\hat\theta_{it}^{j}}{\alpha_{it}^{j}}\right],
\qquad j \in \{1,2\}.
\]

\paragraph{Algorithm.}
\begin{enumerate}
    \item Estimate the production function and recover $\hat\theta_{it}^{1}$ and $\hat\theta_{it}^{2}$ using the researcher's preferred estimator.
    \item Construct the revenue shares $\alpha_{it}^{1}$ and $\alpha_{it}^{2}$ and compute the input-specific estimated markups $\hat\mu_{it}^{1}$ and $\hat\mu_{it}^{2}$.
    \item For each period $t$, compute the sample analogs $\hat s_t^{1}$, $\hat s_t^{2}$, $\hat m_t^{1}$, and $\hat m_t^{2}$.
    \item Compute the regime-selection differences
    \[
    \hat\Delta s_t^{1,2} = \hat s_t^{2} - \hat s_t^{1},
    \qquad
    \hat\Delta m_t^{1,2} = \hat m_t^{1} - \hat m_t^{2}.
    \]
    \item Apply the piecewise FAS formula to obtain the estimated lower and upper FAS endpoints, $\widehat{FAS}_{lb,t}^{1,2}$ and $\widehat{FAS}_{ub,t}^{1,2}$.
    \item Draw bootstrap samples (e.g. clustered by firm), and in each bootstrap replication re-estimate the production function, recompute the revenue shares and the year-level moments, and then evaluate the estimated directional derivative corresponding to the sample regime.
    \item Use the empirical bootstrap distribution of the directional-derivative statistics to construct bootstrap critical values for the lower and upper FAS endpoints.
    \item Construct one-sided or two-sided confidence intervals for each scalar FAS endpoint.
\end{enumerate}

\paragraph{Pseudo-Stata code.}
\begin{verbatim}
/********************************************************************/
/* PSEUDO-STATA: FAS FOR THE AVERAGE MARKUP WITH CONFIDENCE INTERVALS */
/* Two-input case: j = 1, 2                                          */
/********************************************************************/

*---------------------------------------------------------------*
* INPUTS REQUIRED
*---------------------------------------------------------------*
* Firm-year panel data with:
*   - firm identifier:     id
*   - time identifier:     t
*   - total revenue:       sales
*   - expenditure on input 1: exp1
*   - expenditure on input 2: exp2
*   - estimated elasticity for input 1: theta1_hat
*   - estimated elasticity for input 2: theta2_hat
*
* The researcher must obtain theta1_hat and theta2_hat from a
* chosen production-function estimator before running the steps below.
*
* The baseline markup formula is mu_it^j = theta_it^j / alpha_it^j,
* where alpha_it^j is the revenue share of variable input j.
*---------------------------------------------------------------*

*---------------------------------------------------------------*
* STEP 0. CONSTRUCT REVENUE SHARES AND INPUT-SPECIFIC MARKUPS
*---------------------------------------------------------------*
gen alpha1 = exp1 / sales
gen alpha2 = exp2 / sales
gen shareinv1 = 1 / alpha1
gen shareinv2 = 1 / alpha2
gen markup1_hat = theta1_hat / alpha1
gen markup2_hat = theta2_hat / alpha2

*---------------------------------------------------------------*
* STEP 1. COMPUTE YEAR-SPECIFIC SAMPLE ANALOGS OF THE MOMENTS
*         s_t^j = E[1/alpha_it^j]
*         m_t^j = E[theta_hat_it^j / alpha_it^j]
*---------------------------------------------------------------*
bysort t: egen s1_hat_t = mean(shareinv1)
bysort t: egen s2_hat_t = mean(shareinv2)
bysort t: egen m1_hat_t = mean(markup1_hat)
bysort t: egen m2_hat_t = mean(markup2_hat)
bysort t: gen one = 1
bysort t: egen n_t = total(one)
bysort t: keep if _n == 1

*---------------------------------------------------------------*
* STEP 2. COMPUTE THE REGIME-SELECTION DIFFERENCES
*         Delta s_t^(1,2) = s_t^2 - s_t^1
*         Delta m_t^(1,2) = m_t^1 - m_t^2
*---------------------------------------------------------------*
gen ds_hat = s2_hat_t - s1_hat_t
gen dm_hat = m1_hat_t - m2_hat_t

*---------------------------------------------------------------*
* STEP 3. COMPUTE THE SAMPLE FAS ENDPOINTS
*---------------------------------------------------------------*
gen FAS_lb_hat = .
replace FAS_lb_hat = m2_hat_t                  if ds_hat <= 0 & dm_hat >= 0
replace FAS_lb_hat = (m1_hat_t + m2_hat_t)/2  if ds_hat <= 0 & dm_hat <  0
replace FAS_lb_hat = (m1_hat_t + m2_hat_t)/2  if ds_hat >  0 & dm_hat >= 0
replace FAS_lb_hat = m1_hat_t                  if ds_hat >  0 & dm_hat <  0
gen FAS_ub_hat = .
replace FAS_ub_hat = (m1_hat_t + m2_hat_t)/2  if ds_hat <= 0 & dm_hat >= 0
replace FAS_ub_hat = m2_hat_t                  if ds_hat <= 0 & dm_hat <  0
replace FAS_ub_hat = m1_hat_t                  if ds_hat >  0 & dm_hat >= 0
replace FAS_ub_hat = (m1_hat_t + m2_hat_t)/2  if ds_hat >  0 & dm_hat <  0

* Save this dataset as the original year-level statistics
save original_year_stats, replace

*---------------------------------------------------------------*
* STEP 4. BOOTSTRAP THE FIRST-STAGE ESTIMATOR
* In each bootstrap replication b:
*   (i) resample firms (or clusters) with replacement;
*   (ii) re-estimate the production function;
*   (iii) recompute theta1_hat_b and theta2_hat_b;
*   (iv) recompute alpha1_b, alpha2_b, markup1_hat_b, markup2_hat_b;
*   (v) recompute year-level moments.
*---------------------------------------------------------------*
tempfile bootstats
tempname posth
postfile `posth' int b int t double Z_lb_b Z_ub_b using `bootstats', replace
forvalues b = 1/BOOT_REPS {
    preserve
        * 4.1 Draw a bootstrap sample
        bsample, cluster(id)
        * 4.2 Re-estimate the production function
        * User-supplied routine:
        *   produce theta1_hat_b and theta2_hat_b
        * Example placeholder:
        * do "estimate_production_function.do"
        * 4.3 Recompute shares and input-specific markups
        gen alpha1_b = exp1 / sales
        gen alpha2_b = exp2 / sales
        gen shareinv1_b = 1 / alpha1_b
        gen shareinv2_b = 1 / alpha2_b
        gen markup1_hat_b = theta1_hat_b / alpha1_b
        gen markup2_hat_b = theta2_hat_b / alpha2_b
        * 4.4 Recompute year-level sample analogs
        bysort t: egen s1_hat_tb = mean(shareinv1_b)
        bysort t: egen s2_hat_tb = mean(shareinv2_b)
        bysort t: egen m1_hat_tb = mean(markup1_hat_b)
        bysort t: egen m2_hat_tb = mean(markup2_hat_b)
        bysort t: gen one_b = 1
        bysort t: egen n_t_b = total(one_b)
        bysort t: keep if _n == 1
        * 4.5 Merge with original year-level statistics
        merge 1:1 t using original_year_stats, nogen keep(3)
        * 4.6 Construct bootstrap perturbations
        gen g_m1 = sqrt(n_t) * (m1_hat_tb - m1_hat_t)
        gen g_m2 = sqrt(n_t) * (m2_hat_tb - m2_hat_t)
        * 4.7 Estimated directional derivative: lower endpoint
        gen Z_lb_b = .
        replace Z_lb_b = g_m2                if ds_hat < 0 & dm_hat > 0
        replace Z_lb_b = (g_m1 + g_m2)/2    if ds_hat < 0 & dm_hat < 0
        replace Z_lb_b = (g_m1 + g_m2)/2    if ds_hat > 0 & dm_hat > 0
        replace Z_lb_b = g_m1                if ds_hat > 0 & dm_hat < 0
        * 4.8 Estimated directional derivative: upper endpoint
        gen Z_ub_b = .
        replace Z_ub_b = (g_m1 + g_m2)/2    if ds_hat < 0 & dm_hat > 0
        replace Z_ub_b = g_m2                if ds_hat < 0 & dm_hat < 0
        replace Z_ub_b = g_m1                if ds_hat > 0 & dm_hat > 0
        replace Z_ub_b = (g_m1 + g_m2)/2    if ds_hat > 0 & dm_hat < 0
        * 4.9 Store bootstrap statistics by year
        quietly {
            forvalues yy = START_YEAR/END_YEAR {
                count if t == `yy'
                if r(N) == 1 {
                    su Z_lb_b if t == `yy', meanonly
                    local zlb = r(mean)
                    su Z_ub_b if t == `yy', meanonly
                    local zub = r(mean)
                    post `posth' (`b') (`yy') (`zlb') (`zub')
                }
            }
        }
    restore
}
postclose `posth'

*---------------------------------------------------------------*
* STEP 5. COMPUTE BOOTSTRAP QUANTILES
* For any tau in (0,1):
*   c_lb,tau = tau-quantile of Z_lb_b
*   c_ub,tau = tau-quantile of Z_ub_b
*---------------------------------------------------------------*
use `bootstats', clear
preserve
    keep t Z_lb_b
    statsby c_lb_95 = r(c_1) c_lb_05 = r(c_2), by(t) ///
        saving(q_lb, replace): centile Z_lb_b, centile(95 5)
restore
preserve
    keep t Z_ub_b
    statsby c_ub_95 = r(c_1) c_ub_05 = r(c_2), by(t) ///
        saving(q_ub, replace): centile Z_ub_b, centile(95 5)
restore

*---------------------------------------------------------------*
* STEP 6. CONSTRUCT CONFIDENCE INTERVALS FOR THE FAS ENDPOINTS
*---------------------------------------------------------------*
use original_year_stats, clear
merge 1:1 t using q_lb, nogen
merge 1:1 t using q_ub, nogen
* One-sided lower CI for the lower endpoint:
*   [ FAS_lb_hat - c_lb,1-alpha / sqrt(n), +infty )
gen CI_lb_lower = FAS_lb_hat - c_lb_95 / sqrt(n_t)
* One-sided upper CI for the upper endpoint:
*   ( -infty, FAS_ub_hat - c_ub,alpha / sqrt(n) ]
gen CI_ub_upper = FAS_ub_hat - c_ub_05 / sqrt(n_t)
* Optional: two-sided CIs
gen CI_lb_lo_2s = FAS_lb_hat - c_lb_95 / sqrt(n_t)
gen CI_lb_hi_2s = FAS_lb_hat - c_lb_05 / sqrt(n_t)
gen CI_ub_lo_2s = FAS_ub_hat - c_ub_95 / sqrt(n_t)
gen CI_ub_hi_2s = FAS_ub_hat - c_ub_05 / sqrt(n_t)

*---------------------------------------------------------------*
* STEP 7. REPORT
* Report by year:
*   - FAS_lb_hat, FAS_ub_hat
*   - CI_lb_lower, CI_ub_upper
*   - optionally plot the FAS over time with its endpoint CIs
*---------------------------------------------------------------*
\end{verbatim}

\paragraph{Implementation notes.}
First, the pseudo-code above is intentionally agnostic about the first-stage production-function estimator. In applications, the researcher should replace the placeholder in Step 4.2 with the chosen production-function routine and then feed the resulting elasticity estimates into the remaining steps. Second, the code is specialized to the two-input case discussed in the main text. With more than two variable inputs, one may either aggregate inputs into economically meaningful bundles or compute all pairwise FAS objects and report their union, as discussed in Appendix~B. Third, the bootstrap confidence intervals rely on the local regime-separation condition, namely $\Delta s_t^{1,2} \neq 0$ and $\Delta m_t^{1,2} \neq 0$, together with a valid Gaussian limit law and bootstrap consistency for the  estimator.

\end{document}